\documentclass[11pt]{article}
\usepackage[utf8]{inputenc}
\usepackage{amssymb,amsmath,amsthm}
\usepackage{complexity}
\usepackage{xcolor}
\usepackage[colorlinks=true, allcolors=blue]{hyperref}
\usepackage{mathtools}
\usepackage{enumerate}
\usepackage[capitalize]{cleveref}
\usepackage{graphicx}
\usepackage{dsfont}
\usepackage{todonotes}
\usepackage{bm}
\usepackage{fullpage}
\usepackage{thm-restate}

\usepackage{amsthm}
\usepackage{thmtools,thm-restate}

\numberwithin{equation}{section}

\declaretheoremstyle[bodyfont=\it,qed=\qedsymbol]{noproofstyle}

\declaretheorem[name=Observation,numbered=no]{observation*}

\declaretheorem[numberlike=equation]{theorem}

\declaretheorem[name=Theorem,numbered=no]{theorem*}

\declaretheorem[numberlike=equation]{lemma}
\declaretheorem[name=Lemma,numbered=no]{lemma*}

\declaretheorem[name=Corollary,numbered=no]{corollary*}

\declaretheorem[numberlike=equation]{proposition}
\declaretheorem[name=Proposition,numbered=no]{proposition*}

\declaretheorem[name=Claim,numbered=no]{claim*}

\declaretheorem[name=Conjecture,numbered=no]{conjecture*}

\declaretheorem[name=Question,numbered=no]{question*}

\declaretheoremstyle[bodyfont=\it]{defstyle} 

\declaretheorem[numberlike=equation,style=defstyle]{definition}
\declaretheorem[unnumbered,name=Definition,style=defstyle]{definition*}

\declaretheorem[unnumbered,name=Notation=defstyle]{notation*}

\declaretheorem[unnumbered,name=Construction,style=defstyle]{construction*}

\declaretheoremstyle[]{rmkstyle}

\declaretheorem[unnumbered,name=Example,style=rmkstyle]{example*}

\usepackage{algorithmicx}
\usepackage{algorithm}
\usepackage[noend]{algpseudocode}
\newcounter{algsubstate}
\renewcommand{\thealgsubstate}{\alph{algsubstate}}

\algnewcommand\algorithmicinput{\textbf{Input:}}
\algnewcommand\Input{\item[\algorithmicinput]}

\algnewcommand\algorithmicoutput{\textbf{Output:}}
\algnewcommand\Output{\item[\algorithmicoutput]}

\algnewcommand\algorithmicgoal{\textbf{Goal:}}
\algnewcommand\Goal{\item[\algorithmicgoal]}

\DeclarePairedDelimiter{\ceil}{\lceil}{\rceil}
\DeclarePairedDelimiter\floor{\lfloor}{\rfloor}

\newcommand{\eps}{\varepsilon}
\renewcommand{\epsilon}{\varepsilon}
\newcommand{\mcsp}{\textsf{Max-CSP}}
\newcommand{\Exp}{\mathop{\mathbb{E}}}

\newcommand{\cA}{\mathcal{A}}

\newcommand{\cD}{\mathcal{D}}
\newcommand{\cF}{\mathcal{F}}

\newcommand{\N}{\mathbb{N}}

\renewcommand{\R}{\mathbb{R}}
\newcommand{\bern}{\mathsf{Bern}}

\newcommand{\maxf}{\textsf{Max-CSP}(f)}
\newcommand{\bias}{\textsf{bias}}
\newcommand{\val}{\textsf{val}}
\newcommand{\sign}{\textsf{sign}}

\newcommand{\mon}{\textsf{MON}}
\newcommand{\maj}{\textsf{MAJ}}
\newcommand{\wmon}{\textsf{WMON}}

\newcommand{\one}{\mathds{1}}
\newcommand{\veca}{\mathbf{a}}
\newcommand{\vecb}{\mathbf{b}}

\newcommand{\vecj}{\mathbf{j}}

\newcommand{\vecu}{\mathbf{u}}
\newcommand{\vecv}{\mathbf{v}}

\newcommand{\vecx}{\mathbf{x}}
\newcommand{\veccx}{\mathbf{X}}
\newcommand{\vecy}{\mathbf{y}}

\newcommand{\vecsigma}{\boldsymbol{\sigma}}
\newcommand{\veclambda}{\boldsymbol{\lambda}}
\newcommand{\vecmu}{\boldsymbol{\mu}}

\title{Sketching Approximability of (Weak) Monarchy Predicates}
\author{
Chi-Ning Chou\thanks{School of Engineering and Applied Sciences, Harvard University, Cambridge, Massachusetts, USA. Partially supported by NSF grants DMS-2134157 and CCF-1565264, DOE grant DE-SC0022199, and the Simons foundation. Email: \texttt{chiningchou@g.harvard.edu}.}
\and Alexander Golovnev\thanks{Department of Computer Science, Georgetown University. Email: \texttt{alexgolovnev@gmail.com}.}
\and Amirbehshad Shahrasbi\thanks{Microsoft, USA. Email: \texttt{ashahrasbi@microsoft.com}. The author was with Harvard University and supported by CRA-CCC Computing Innovations Fellowship (CIFellowship 2020) during this work.}
\and Madhu Sudan\thanks{School of Engineering and Applied Sciences, Harvard University, Cambridge, Massachusetts, USA. Supported in part by a Simons Investigator Award and NSF Awards CCF 1715187 and CCF 2152413. Email: \texttt{madhu@cs.harvard.edu}.}
\and Santhoshini Velusamy\thanks{School of Engineering and Applied Sciences, Harvard University, Cambridge, Massachusetts, USA. Supported in part by a Google Ph.D. Fellowship, a Simons 
Investigator Award to Madhu Sudan, and NSF Awards CCF 1715187 and CCF 2152413. Email: \texttt{svelusamy@g.harvard.edu}.}
}
\date{}

\allowdisplaybreaks
\begin{document}

\maketitle

\begin{abstract}
    We analyze the sketching approximability of constraint satisfaction problems on Boolean domains, where the constraints are balanced linear threshold functions applied to literals. In~particular, we explore the approximability of monarchy-like functions where the value of the function is determined by a weighted combination of the vote of the first variable (the president) and the sum of the votes of all remaining variables. The pure version of this function is when the president can only be overruled by when all remaining variables agree. For every $k \geq 5$, we show that CSPs where the underlying predicate is a pure monarchy function on $k$ variables have no non-trivial sketching approximation algorithm in $o(\sqrt{n})$ space. We also show infinitely many weaker monarchy functions for which CSPs using such constraints are non-trivially approximable by $O(\log(n))$ space sketching algorithms. Moreover, we give the first example of sketching approximable asymmetric Boolean CSPs.
    Our results work within the framework of Chou, Golovnev, Sudan, and Velusamy (FOCS 2021) that characterizes the sketching approximability of all CSPs. Their framework can be applied naturally to get a computer-aided analysis of the approximability of any specific constraint satisfaction problem. The novelty of our work is in using their work to get an analysis that applies to \emph{infinitely} many problems simultaneously.
\end{abstract}
\thispagestyle{empty}
\newpage
\setcounter{page}{1}

\section{Introduction}
In this paper we consider the sketching complexity of solving constraint satisfaction problems (CSPs) approximately where the constraints are given by linear threshold functions over a collection of Boolean literals. We introduce these terms below.

\paragraph{CSPs:} Given a Boolean function $f:\{-1,1\}^k \to \{0,1\}$, the Boolean CSP associated with $f$, denoted $\maxf$ is the following optimization problem. Given $m$ constraints $C_1,\ldots,C_m$ on $n$ Boolean variables $X_1,\ldots,X_n$, where each constraint applies $f$ to a sequence of $k$ distinct literals from the set $\{X_1,\ldots,X_n, -X_1,\ldots,-X_n\}$, find the maximum fraction of constraints that can be satisfied by an assignment to the $n$ variables. For an instance $\Psi$ of $\maxf$ we use $\val_\Psi$ to denote this maximum value.
We are interested in approximating $\val_\Psi$ and this task is known to be equivalent to solving a gapped decision version of $\maxf$. For $0 \leq \beta < \gamma \leq 1$ we define the $(\gamma,\beta)$-gapped version of $\maxf$, abbreviated to $(\gamma,\beta)$-$\maxf$, to be the following promise decision problem: Given an instance $\Psi$ satisfying $\val_\Psi\geq \gamma$ or $\val_\Psi < \beta$ decide which one of the two conditions holds.

\paragraph{Sketching algorithms:} The class of algorithms we consider (and rule out) are randomized sketching algorithms. Inputs to these algorithms arrive as a stream of elements, in our case a stream of constraints. We consider algorithms that use some bounded amount of space, denoted $s(n)$, to process the stream and maintain a sketch of their output. When the stream ends the algorithm outputs it verdict based on the current sketch. A key restriction of a sketching algorithm is that its sketch should satisfy the following composability property. Given  two streams $\sigma$ and $\tau$ and a fixing of the randomness, the sketch of their concatenation $S(\sigma \circ \tau)$ should be determined by their sketches $S(\sigma)$ and $S(\tau)$ alone.\footnote{In contrast, a general streaming algorithm maintains a state $S(\sigma \circ \tau)$ that may depend on $S(\sigma)$ and all of $\tau$.} Most existing algorithms for streaming CSPs are sketching algorithms. We say a sketching algorithm solves a (gapped) decision problem if on every input its answer is correct with probability at least $2/3$. 

\paragraph{Approximability and approximation resistance:} 
For $\alpha\in[0,1]$, we say an algorithm is an $\alpha$-approximation algorithm for $\maxf$ if the following holds: on every input instance $\Psi$, the algorithm outputs $v$ such that $\alpha\cdot\val_\Psi\leq v\leq\val_\Psi$ with probability at least $2/3$. Note that the existence of an $\alpha$-approximation algorithm is equivalent to the existence of an algorithm for solving $(\gamma,\beta)$-$\maxf$ for every $\gamma,\beta\in[0,1]$ with $\beta\le \alpha\cdot\gamma$.

For a function $f:\{-1,1\}^k \to \{0,1\}$, define $\rho(f) = 2^{-k}\cdot | \{x \in \{-1,1\}^k | f(x) = 1\}|$.
For every $f$ and every instance $\Psi$ of $\maxf$, a random assignment satisfies $\rho(f)$ fraction of the constraints in expectation and so every $\Psi$ satisfies $\val_\Psi\geq \rho(f)$. Thus the $(1,\rho(f))$-$\maxf$ problem is trivially solvable by the algorithm that always outputs $\val_{\Psi} \geq 1$ (since the set $\{\Psi | \val_\Psi < \rho(f)\}$ is empty). We say $\maxf$ is sketching approximable within space $s(n)$ if there is an $\epsilon > 0$ and a sketching algorithm using at most $s(n)$ space that solves $(1-\epsilon,\rho(f)+\epsilon)$-$\maxf$. We say that $\maxf$ is approximation resistant to space $s(n)$ if for every $\epsilon > 0$, every sketching algorithm for $(1,\rho(f)+\epsilon)$-$\maxf$ requires $\Omega(s(n))$ space.

\subsection{Motivation and related work}
There has been an increasing interest in studying the approximability of CSPs in the streaming setting~\cite{KK15,KKS15,KKSV17,GVV17,GT19,KK19,CGV20,CGSV21-finite,CGSV21-boolean,SSV21,BHP+22,CGS+22}.
In particular, recently Chou, Golovnev, Sudan, and Velusamy~\cite{CGSV21-finite,CGSV21-boolean} gave a dichotomy result for sketching approximability of all finite CSPs. Specifically, they proved the following theorem.

\begin{theorem}[\cite{CGSV21-boolean}] For every $k$, every predicate $f:\{-1,1\}^k \to \{0,1\}$ and every $0 \leq \beta < \gamma \leq 1$ one of the following holds: (1) $(\gamma,\beta)$-$\maxf$ is solvable by an $O(\log(n))$-space sketching algorithm, or (2) for every $\epsilon > 0$, $(\gamma-\epsilon,\beta+\epsilon)$-$\maxf$ is not solvable by any $o(\sqrt{n})$-space sketching algorithm. Furthermore there is a decidable procedure that determines, given $\cF$, $\gamma$ and $\beta$, which of the two conditions hold. 
\end{theorem}

We note that a followup paper by the same authors~\cite{CGSV21-finite} extends the result to a more general setting: Specifically they allow non-Boolean variables, allow a set of predicates rather than a single function; and allow the predicates to be applied to variables rather than literals. While their result is more general all results in this paper work in the more restricted setting of \cite{CGSV21-boolean} and so we will describe our results in their language (which can be somewhat simpler for problems that are expressible in their setting). 

While the results of \cite{CGSV21-boolean} imply a dichotomy, to 
explicitly get the optimal sketching approximation ratio for a given predicate $f$, they need to solve an optimization problem which in general needs computer-aided analysis.
In order to get more explicit results one needs to restrict the families of functions considered, and even then it is unclear if there can be a closed-form expression. In the only example we are aware of, Boyland, Hwang, Prasad, Singer, and Velusamy~\cite{BHP+22} gave closed-form expressions for the optimal sketching approximation ratio of some {\em symmetric} Boolean CSPs. This still leaves the question of  exploring the sketching approximability of other subfamilies of CSPs and extracting some qualitative results yielding necessary or sufficient conditions for non-trivial approximability.

\subsection{Main results}
In this paper we study sketching approximability of CSPs on linear threshold functions. Below we define the classes of linear threshold functions and balanced linear threshold functions.
 
\begin{definition}[Linear threshold function]
A linear threshold function, or LTF, is a Boolean function $f:\{-1,1\}^k\rightarrow\{0,1\}$ of the form \[f(x) = \sign\left(\sum_{i=1}^k w_i x_i + \theta \right),\] where $w_1,\dots,w_k,\theta\in \mathbb{R}$. The function $\sign(z)$ has value $1$ if $z> 0$ and $0$ if $z\le 0$; $w_1,\dots,w_k$ are called the weights of $f$ and $\theta$ is the threshold. 
\end{definition}

\begin{definition}[Balanced linear threshold function]
A balanced linear threshold function, or balanced LTF, is an LTF with threshold $0$ and the additional restriction that for every $x \in \{-1,1\}^k$, we have $\sum_{i=1}^k w_i x_i \ne 0$. Specifically, a balanced LTF $f$ satisfies $f(-x) = 1-f(x)$ for every $x$.
\end{definition}

Note that for a balanced LTF~$f$, $\rho(f) = 1/2$, and the goal of approximability is to beat this factor. Balanced LTFs form a technically important class of functions to study visavis CSP approximability. For instance Potechin~\cite{Pot19} studies them in the polynomial time regime giving a (somewhat complex) approximation-resistant function in this class. 
In the sketching setting, interest in this class of functions comes from 
 \cite[Theorem 1.3]{CGSV21-boolean} which shows that if a function $f$ supports one-wise independence (i.e., $f^{-1}$ supports a distribution on $\{-1,1\}^k$ that is uniform on each of the $k$ marginals) then $\maxf$ is approximation resistant to $o(\sqrt{n})$ space streaming algorithms. Balanced LTFs are the most basic class of functions that \emph{do not} support one-wise independence and hence are not covered by this theorem. Studying this class thus offers the possibility of finding new classes of CSPs that are approximation resistant to $o(\sqrt{n})$-space streaming algorithms. 

Our first result shows that every balanced LTF on up to $4$ variables is sketching approximable. (So to search for new approximation resistant functions we need to look at functions on more variables!) We note that there are only finitely many such LTFs, but already this theorem gives the first example of an asymmetric Boolean CSP which is approximable by sketching algorithms.\footnote{Note that \textsf{Max-DICUT} (shown to be sketching approximable in~\cite{CGV20,CGSV21-finite}) is not considered a Boolean CSP in \cite{CGSV21-boolean} since the \textsf{Max-DICUT} constraints are applied on variables and not on literals.}
\begin{restatable}{theorem}{thmthree}\label{thm:4ltfs}
For every balanced LTF $f$ on $k\leq 4$ variables, $\maxf$ is sketching approximable in $O(\log(n))$ space.
\end{restatable}

Our next result shows that there do exist balanced LTFs functions on $5$ or more variables that are sketching approximation resistant. The specific family of functions we show this for are the ``Monarchy'' functions. For $k \in \N$, $\mon_k:\{-1,1\}^k \to \{0,1\}$ is given by $\mon_k(x_1,\ldots,x_k) = \sign\left((k-2)x_1+ x_2 + \cdots + x_k \right)$. It may be easily verified that $\mon_k$ is a balanced LTF. We have the following theorem.

\begin{restatable}{theorem}{thmone}\label{thm:main}
For every $k\geq5$, $\textsf{Max-CSP}(\mon_k)$ is sketching approximation resistant to space $o(\sqrt{n})$.
\end{restatable}

Thus we get the first examples of functions that do not support one-wise independence that are approximation resistant to space $o(\sqrt{n})$ sketching algorithms. In fact, the theorem gives infinitely many such examples.  We suspect that the Balanced LTF constructed in \cite{Pot19} should also be approximation-resistant but so far we don't have a proof. The monarchy functions, by virtue of the simplicity allow a simpler analytic proof, though admittedly even in this case we do not have great intuition for the proof and do not know how to extend it to other classes of functions. 

Finally we also give an infinite subclass of balanced LTFs that are approximable using $O(\log(n))$ space. The functions we consider here are what we call ``weak monarchy'' functions.\footnote{Such functions are also sometimes called presidential type predicates~\cite{hp20}.} For $j \leq k \in \N$, let $\wmon_{k,j}:\{-1,1\}^k \to \{0,1\}$ be the function given by $\wmon_{k,j}(x_1,\ldots,x_k) = \sign\left(j\cdot x_1+ x_2 + \cdots + x_k \right)$. It may be easily verified that when $j+k$ is even, then $\wmon_{k,j}$ is a balanced LTF. We have
\begin{restatable}{theorem}{thmtwo}\label{thm:wmon-apx}
For all integers $j\geq 2$ and $k\geq7j^3$ such that $k+j$ is even, $\textsf{Max-CSP}(\wmon_{k, j})$ is sketching approximable in $O(\log(n))$ space.
In particular, for every $j$, there exist infinitely many $k$ such that $\textsf{Max-CSP}(\wmon_{k,j})$ is sketching approximable.
\end{restatable}
The results above give the first examples of asymmetric Boolean CSPs for which $\maxf$ is sketching approximable. Again we get an infinite family of such functions.

\paragraph{Comparison to the polynomial time regime.} Hast~\cite{h05} proves that (a generalization of) \cref{lem:approximablefn} holds in the polynomial time regime (thus, implying an analogue of \cref{thm:wmon-apx} in the polynomial time regime). Austrin, Benabbas, and Magen~\cite{abm10} prove that $\mon_k$ is approximable in polynomial time, which is in sharp contrast to the result of \cref{thm:main} in the sketching setting. Huang and Potechin~\cite{hp20} show that almost all $\wmon$ predicates are approximable in polynomial time. Finally, Potechin~\cite{Pot19} gives a balanced LTF which is (conditionally) approximation resistant in the polynomial time regime.

\paragraph{Organization of the paper.}
We start with giving formal definitions and stating relevant previous results in~\cref{sec:prelim}. The three main theorems are proved in~\cref{sec:additional},~\cref{sec:proof of main thm}, and~\cref{sec:Theorem2_proof}, respectively.

\section{Preliminaries}\label{sec:prelim}

We use $\N,\R$, and $\R_{\geq0}$ to denote the sets of all natural, real, and non-negative real numbers, respectively. We use $[n]$ to denote the set $\{1,\ldots,n\}$. We write vector variables in boldface, e.g., $\vecx$, and we use $x_i$ to denote their $i$th entry. For two vectors of the same length $\vecx,\vecy\in\R^k$, $\vecx\odot\vecy\in\R^k$ denotes the entry-wise product of $\vecx$ and $\vecy$. For $p\in[0,1]$, $\bern(p)$ denotes the Bernoulli distribution taking value~$1$ with probability~$p$, and value~$-1$ with probability~$1-p$. We adopt the convention that $\binom{n}{k}=0$ for $k<0$ or $k>n$. By $\binom{n}{\leq k}$ we denote the sum $\sum_{i=0}^{k}\binom{n}{i}$.
\subsection{Sketching approximability and approximation resistance}
For a function $f\colon\{-1,1\}^k\to\{0,1\}$, let $\rho(f)=2^{-k}\cdot|\{\veca\in\{-1,1\}^k\, |\, f(\veca)=1\}|$ denote the probability that a uniformly random assignment of the variables satisfies~$f$. 
\begin{definition}[Sketching approximation resistance]
For a function $f\colon\{-1,1\}^k\to\{0,1\}$, we say that $f$ is \emph{sketching approximation resistant} to space $s(n)$ if for every $\epsilon > 0$, every sketching algorithm for $(1,\rho(f)+\epsilon)$-$\maxf$ requires $\Omega(\sqrt{n})$ space.
\end{definition}

\begin{definition}[Sketching approximability]
For a function $f\colon\{-1,1\}^k\to\{0,1\}$, we say that $f$ is \emph{sketching approximable} in space $s(n)$ if there exist $\epsilon > 0$ and a sketching algorithm that solves $(1-\epsilon,\rho(f)+\epsilon)$-$\maxf$ using space $s(n)$.
\end{definition}

At first glance, it seems that if $f$ is not sketching approximation resistant then it's not necessarily sketching approximable. Nonetheless, \cite{CGSV21-boolean} proved that every $f$ is either approximable or approximation resistant.\footnote{Concretely, as the sets $K^Y,K^N$ are closed (see~\cref{lem:convex}), an algorithm for $(1,\rho(f)+\epsilon)$-$\maxf$ implies an algorithm for $(1-\epsilon',\rho(f)+\epsilon)$-$\maxf$ for some $\epsilon'>0$, which in turn implies that $\maxf$ is approximable.}

\subsection{Characterization of approximability from~\texorpdfstring{\cite{CGSV21-boolean}}{[CGSV21]}}
In this work, we focus on CSPs that use a single function $f$ applied to literals. Thus, we will use the machinery from~\cite{CGSV21-boolean} instead of the more general (and more notationally-heavy) version in~\cite{CGSV21-finite}.
For a distribution $\cD \in \Delta(\{-1,1\}^k)$, by $\vecmu(\cD)$ we denote its marginals, i.e., $\vecmu(\cD) = (\mu_1,\ldots,\mu_k)$ where $\mu_i = \Exp_{\vecb\sim\cD}[b_i]$ for all $i\in[k]$.
\begin{definition}[{\cite[Definitions~2.1 and~2.2]{CGSV21-boolean}}]\label{def:sets}
For $\gamma,\beta \in \mathbb{R}$, we define the sets of distributions $S_\gamma^Y$ and $S_\beta^N$ as
\begin{align*}
S_\gamma^Y  = S_\gamma^Y(f) & = \{\cD_Y \in \Delta(\{-1,1\}^k) ~\vert ~\Exp_{\vecb\sim \cD_Y}[f(\vecb)]\ge \gamma \}
\intertext{and}
S_\beta^N = S_\beta^N(f) & = \{\cD_N\in \Delta(\{-1,1\}^k) ~\vert ~\Exp_{\vecb\sim \cD_N}\Exp_{\veca \sim \bern(p)^k}[f(\vecb \odot \veca)]\leq \beta, \forall p\in[0,1] \} \, ,
\end{align*}
and the sets of marginals of these distributions
\begin{align*}
K_\gamma^Y = K_\gamma^Y(f) & = \{~\vecmu(\cD_Y)~\vert ~ \cD_Y \in S_\gamma^Y \}
\intertext{and}
K_\beta^N = K_\beta^N(f)& = \{~\vecmu(\cD_N)~\vert ~ \cD_N \in S_\beta^N \} \, .
\end{align*}
\end{definition}

We will use the following properties of the sets $K_\gamma^Y$ and $K_\beta^N$.
\begin{lemma}[{\cite[Lemma~2.4]{CGSV21-boolean}}]\label{lem:convex}
For every $\gamma,\beta \in [0,1]$ the sets $K_\gamma^N$ and $K_\beta^Y$ are bounded, closed and convex.
\end{lemma}

With these definitions, we are ready to present the approximability criteria from~\cite{CGSV21-boolean}.\footnote{Strictly speaking the statement in Corollary 1.2 in \cite{CGSV21-boolean} is somewhat different, but their proof of Corollary~1.2 asserts this explicitly.}
\begin{theorem}[{\cite[Corollary~1.2]{CGSV21-boolean}}]\label{thm:cgsv}
For every $k\in\N$ and every function $f:\{-1,1\}^k \to \{0,1\}$, if $K_1^Y(f) \cap K_{\rho(f)}^N(f) = \emptyset$, then $f$ is sketching approximable within space $O(\log(n))$, if $K_1^Y(f) \cap K_{\rho(f)}^N(f) \neq \emptyset$, then $f$ is sketching approximation resistant to space $o(\sqrt{n})$.
\end{theorem}

\subsection{(Weak) Monarchy functions}

\begin{definition}
A monarchy predicate on $k\geq2$ variables $\mon_k\colon\{-1,1\}^k\to\{0,1\}$ is defined as
\[
\mon_k(x_1,\ldots,x_k) = \sign\left((k-2)x_1 + \sum_{i=2}^k x_i\right)
\;.
\]
Here $x_1$ is commonly referred to as the president and the rest of $x_i$s are called citizens.
\end{definition}

\begin{definition}[Weak monarchy functions]
A weak monarchy predicate of order $j$ on $k\geq2$ variables $\wmon_{k, j}\colon\{-1,1\}^k\to\{0,1\}$ is defined as
\[
\wmon_{k, j}(x_1,\ldots,x_k) = \sign\left(j\cdot x_1 + \sum_{i=2}^k x_i\right)
\;.
\]
Similar to ordinary monarchy functions, $x_1$ is commonly referred to as the president and the rest of $x_i$s are called citizens.
\end{definition}

It is straightforward to see that $\mon_k$ is a balanced LTF for every $k\geq2$ and $\wmon_{k,j}$ is a balanced LTF whenever $k+j$ is even.

\subsection{Fourier analysis of Boolean functions}

We will need the following basic notions from Fourier analysis over the Boolean hypercube (see, for instance,~\cite{ryansbook}). 
\begin{definition}[Characteristic functions]
For every $S\subseteq[k]$ such that $|S|\ge 1$, the characteristic function $\chi_S:\{-1,1\}^k\rightarrow \{-1,1\}$ is defined as $\chi_S(x) = \prod_{i\in S} x_i$. The characteristic function corresponding to the empty set is defined as the constant function $\chi_\emptyset(x)=1$ for all $x\in \{-1,1\}^k$.
\end{definition}

\begin{definition}[Fourier expansions]
The Fourier expansion of a Boolean function $f:\{-1,1\}^k\rightarrow \{0,1\}$ is given by 
\[f = \sum_{S\subseteq[k]} \widehat{f}(S) \cdot \chi_S \, ,\]where $ \widehat{f}(S) = \mathbb{E}_{x\sim \mathsf{Unif}\{-1,1\}^k}[f(x)\cdot \chi_S(x)]$ and $\mathsf{Unif}(\{-1,1\}^k)$ denotes the uniform distribution on $\{-1,1\}^k$.
\end{definition}

\begin{definition}[Chow parameters]
The Chow parameters of a Boolean function $f:\{-1,1\}^k\rightarrow \{0,1\}$ are the degree-$0$ Fourier coefficient and the $k$ degree-$1$ Fourier coefficients of $f$, i.e., $\widehat{f}(\emptyset),\widehat{f}(\{1\}),\dots,\widehat{f}(\{k\})$.
\end{definition}

\begin{proposition}\label{prop:Boolean_fn_prop}
For every Boolean function $f:\{-1,1\}^k\rightarrow \{0,1\}$,
\begin{enumerate}
    \item $\rho(f) = \widehat{f}(\emptyset)$,
    \item for every $S\subseteq [k]$, $|\widehat{f}(S)|\le \widehat{f}(\emptyset)$, and
    \item for every $x\in \{-1,1\}^k$, $-\widehat{f}(\emptyset)\cdot k \le \sum_{i=1}^k \widehat{f}(\{i\}) \cdot x_i \le \widehat{f}(\emptyset)\cdot k $.
\end{enumerate}
\end{proposition}

\begin{proof}
The first statement of the proposition follows directly from the definition of $\rho(f)$: $\rho(f)=\mathbb{E}_{x\sim \mathsf{Unif}(\{-1,1\}^k)} [f(x)] =\widehat{f}(\emptyset)$. For the second statement, observe that for all $S\subseteq[k]$,
\begin{align*}
    |\widehat{f}(S)| &= |\mathbb{E}_{x\sim \mathsf{Unif}(\{-1,1\}^k)} [f(x)\cdot \chi_S(x)]| \\
    &\le \mathbb{E}_{x\sim \mathsf{Unif}(\{-1,1\}^k)} [|f(x)\cdot \chi_S(x)|] \\
    &= \mathbb{E}_{x\sim \mathsf{Unif}(\{-1,1\}^k)} [f(x)] \\
    &= \widehat{f}(\emptyset)\, .
\end{align*} It immediately follows that for all $x\in \{-1,1\}^k$, \[\left|\sum_{i=1}^k \widehat{f}(\{i\}) \cdot x_i\right|\le \sum_{i=1}^k |\widehat{f}(\{i\}) \cdot x_i| \le \widehat{f}(\emptyset)\cdot k \, .\]
\end{proof}

\section{Approximability of Balanced LTFs on 4 variables}\label{sec:additional}
In this section, we show that all balanced LTFs on at most~$4$ variables are sketching approximable in $O(\log(n))$ space. We start by proving that $\textsf{Max-CSP}(\mon_4)$ is approximable.

\subsection{Approximability of \texorpdfstring{$\mon_4$}{MON4}}
Recall that by~\cref{thm:cgsv}, it suffices to show that $K^Y_1({\mon_4})\cap K^N_{1/2}(\mon_4)=\emptyset$. For $k\geq2$, the inputs $x_2,\ldots,x_k$ are symmetric, and we will only consider distributions $\cD\in\Delta(\{-1,1\}^k)$ where all vectors having the same sum of coordinates and the same value in the first coordinate have the same probability masses. Concretely, for $\vecx,\vecy\in\{-1,1\}^k$, if $x_1=y_1$ and $\sum_{i}x_i=\sum_{i}y_i$, then $\cD(x)=\cD(y)$. Such a distribution $\cD$ is uniquely specified by a pair of vectors $\vecu=(u_0,\ldots,u_{k-1}), \vecv=(v_0,\ldots,v_{k-1}) \in \R_{\geq0}^k$ with $\sum_{i}u_i+v_i=1$, where for $0\leq i \leq k-1$,

\begin{align*}
u_i &= \Pr\{x_1 = 1 \text{ \;\;\,and exactly $i$ of the rest of $x_i$s are 1}\} \;,\\
v_i &= \Pr\{x_1 = -1 \text{ and exactly $i$ of the rest of $x_i$s are 1}\} \;.
\end{align*}
Note that when $\sum_i u_i+v_i=1$,  $\vecu,\vecv$ define a distribution $\cD$ with marginals $\vecmu(\cD)=(\mu_1,\mu',\dots,\mu')$ where 
\begin{align}\label{eq:mu1}
    \mu_1=\sum_{i=0}^{k-1}(u_i-v_i) \text{ and } \mu'=\sum_{i=0}^{k-1}(\frac{2i}{k-1}-1)(u_i+v_i)\;.
\end{align}

Next we show that for $\mon_k$ functions, restricting our attention to this class of distributions is without loss of generality.
\begin{definition}\label{def:ktilde}
For $\gamma,\beta\in\mathbb{R}$ and $k\geq2$,
\begin{align*}
\widetilde{K}_\gamma^Y(\mon_k) & = \{~(\mu_1,\mu')~\vert ~ (\mu_1,\mu',\ldots,\mu') \in K_\gamma^Y(\mon_k) \}\\
\mbox{ and }\widetilde{K}_\beta^N(\mon_k)& = \{~(\mu_1,\mu')~\vert ~ (\mu_1,\mu',\ldots,\mu') \in K_\beta^N(\mon_k) \} \, .
\end{align*}
\end{definition}
\begin{lemma}\label{lem:intersection}
For $\gamma,\beta\in\mathbb{R}$ and $k\geq2$, 
\[K^Y_\gamma({\mon_k})\cap K^N_\beta(\mon_k)=\emptyset \text{ if and only if } \widetilde{K}^Y_\gamma({\mon_k})\cap \widetilde{K}^N_\beta(\mon_k)=\emptyset \;.
\]
\end{lemma}
\begin{proof}
First, if $(\mu_1,\mu',\ldots,\mu')\in\widetilde{K}^Y_\gamma({\mon_k})\cap \widetilde{K}^N_\beta(\mon_k)$, then by \cref{def:ktilde}, $(\mu_1,\mu',\ldots,\mu')\in{K}^Y_\gamma({\mon_k})\cap {K}^N_\beta(\mon_k)$.

For the other direction. Assume that there is a vector $\vecmu=(\mu_1,\mu_2,\ldots,\mu_k)\in K^Y_\gamma({\mon_k})\cap {K}^N_\beta(\mon_k)$. Consider two distribution $\cD_Y\in S^Y_\gamma$ and $\cD_N\in S^N_\beta$ yielding the vector $\vecmu=\vecmu(\cD_Y)=\vecmu(\cD_N)$. Given that the variables $x_2, \cdots, x_k$ are symmetric, any distribution that is yielded by permuting $x_2, \cdots, x_k$ in $\cD_Y$ (or $\cD_N$) is also in $S^Y_\gamma$ (or $S^N_\beta$). Note that the marginals of these distributions are also permutations of $\vecmu$. By~\cref{lem:convex}, $K^Y_\gamma$ and $K^N_\beta$ are convex, so they also contain the averages of these vectors: $(\mu_1,\mu',\ldots,\mu')\in K^Y_\gamma({\mon_k})\cap {K}^N_\beta(\mon_k)$ for $\mu'=(\mu_2+\ldots+\mu_k)/(k-1)$. Finally, by \cref{def:ktilde}, $(\mu_1,\mu')\in\widetilde{K}^Y_\gamma({\mon_k})\cap \widetilde{K}^N_\beta(\mon_k)$.
\end{proof}
Next, we characterize the set $\widetilde{K}^Y_1(\mon_k)$.
\begin{lemma}\label{lem:ky1}
For every $k\geq2$, $\widetilde{K}^Y_1({\mon_k})=\{(\mu_1,\mu')\in[-1,1]^2\colon \mu_1(k-2)+\mu'(k-1)\geq1\}$.
\end{lemma}
\begin{proof}
For $\mu_1,\mu'\in[-1,1]$ satisfying $\mu_1(k-2)+\mu'(k-1)\geq1$, consider the distribution $\cD_Y$ given by $u_1=\frac{(k-1)(1-\mu')}{2(k-2)}, u_{k-1}=\frac{(k-1)\mu'+(k-2)\mu_1-1}{2(k-2)}, v_{k-1}=(1-\mu_1)/2$, and $u_i=0$ for $i\not\in\{1,k-1\}$ and $v_j=0$ for $j\neq k-1$. Note that $u_1,v_{k-1}\geq0$ from $\mu_1,\mu'\in[-1,1]$, and $u_{k-1}\geq0$ from $\mu_1(k-2)+\mu'(k-1)\geq1$. It is also easy to check that $u_1+u_{k-1}+v_{k-1}=1$ which implies that $\cD_Y$ is a distribution, and that it is supported on the preimages of~$1$ under $\mon_k$. Therefore $(\mu_1,\mu')\in\widetilde{K}^Y_1({\mon_k})$.

For the other direction, a distribution $\cD_Y$ supported on the preimages of~$1$ under $\mon_k$ satisfies ${u_1+\ldots+u_{k-1}+v_{k-1}=1}$. Then, from~\eqref{eq:mu1},
\begin{align*}
    \mu_1(k-2)+\mu'(k-1)
    &=(k-2)\sum_{i=0}^{k-1}(u_i-v_i)+\sum_{i=0}^{k-1}(2i-k+1)(u_i+v_i)\\
    &=\sum_{i=1}^{k-1}(2i-1)u_i+v_{k-1}\\
    &\geq \sum_{i=1}^{k-1}u_i+v_{k-1}
    =1 \;,
\end{align*}
where the second equality uses that $u_0=0$ and $v_j=0$ for $j<k-1$. This concludes the proof of the lemma.
\end{proof}

Now we show that for the $\mon_4$ function, $\widetilde{K}^Y_1$ and $\widetilde{K}^N_{1/2}$ are disjoint, and, thus, $\mon_4$ is approximable in $O(\log(n))$ space.
\begin{lemma}\label{lem:mon4}
$\textsf{Max-CSP}(\mon_4)$ is sketching approximable in $O(\log(n))$ space.
\end{lemma}
\begin{proof}
Note that \Cref{lem:ky1} gives that 
$\widetilde{K}^Y_1({\mon_4})=\{(\mu_1,\mu')\in[-1,1]^2\colon 2\mu_1+3\mu'\geq1\}$. We show that $\widetilde{K}^Y_1$ and $\widetilde{K}^N_{1/2}$ are disjoint, and then \cref{lem:intersection} and \cref{thm:cgsv} imply that $\textsf{Max-CSP}(\mon_4)$ is sketching approximable in space $O(\log(n))$. 
Next, we prove that no distribution $\mathcal{D} \in S^N_{1/2}$ has marginals that lie in $\widetilde{K}^Y_1$. 

We start by characterizing $K^N_{1/2}$ (for general $\mon_k$). Take a distribution $\mathcal{D}\in \Delta(\{-1, 1\}^k)$. In order for $\mathcal{D}$ to lie within $S^N_{1/2}$, the following needs to be satisfied:
\begin{eqnarray}
\Exp_{\vecb\sim \cD_N}\Exp_{\veca \sim \bern(p)^k}[f(\vecb \odot \veca)]\leq \beta, \forall p \;.
\end{eqnarray}

Let the function $h_{\mathcal{D}}(p)$ denote the probability of an assignment from $\mathcal{D}$ that has undergone bit flips with respect to $\bern(p)^k$ to satisfy the monarchy predicate with the probability of $\beta = 1/2$ or less. With this definition, $\mathcal{D} \in S^N_{1/2}$ if and only if 
$h_{\mathcal{D}}(p) \leq \frac{1}{2}$ for all $0\leq p \leq 1.$
Note that negating all variables $x_i$ flips the output of the monarchy predicate. Therefore, the negation of a ``true'' assignment is ``false'' and vice versa. This gives that 
$h_{\mathcal{D}}(p) = 1 - h_{\mathcal{D}}(1 - p)$ for all $0\leq p \leq 1$ which implies that $\mathcal{D} \in S^N_{1/2}$ if and only if for all $0\leq p \leq 1$
\[h_{\mathcal{D}}(p) = \frac{1}{2}\;.\]

We now write down the coefficients of the polynomial $h_{\mathcal{D}}(p)$ in terms of $u_i$ and $v_i$ describing the distribution (as used earlier in this section).

If one draws an assignment from $\mathcal{D}$ where $x_1 = 1$ and exactly $i$ of the rest of the variables are~1, the probability of the resulting assignment satisfying the monarchy predicate after the Bernoulli flipping is 
\[p(1 - (1-p)^ip^{k - 1 - i}) + (1-p)^{k - i} p^i\;.\]
Similarly, if $x_1 = -1$ and exactly $i$ of the rest of the variables are~1, the probability of the resulting assignment satisfying the monarchy predicate after the Bernoulli flipping is \[(1-p)(1 - (1-p)^ip^{k - 1 - i}) + (1-p)^{k - 1 - i}p^{i + 1}\;.\]
This gives that 
\begin{eqnarray}
h_{\mathcal{D}}(p) &=& \sum_{i=0}^{k-1} u_i \left[p(1 - (1-p)^ip^{k - 1 - i}) + (1-p)^{k - i} p^i\right]\nonumber\\
&& +\sum_{i=0}^{k-1} v_i \left[(1-p)(1 - (1-p)^ip^{k - 1 - i}) + (1-p)^{k - 1 - i}p^{i + 1}\right]\label{eqn:hClosedForm}
\end{eqnarray}

To prove this lemma, we form the polynomial $h_{\mathcal{D}}(p)$ for $k=4$ and show that no set of $u_i$s and $v_i$s satisfy both $h_{\mathcal{D}}(p) = \frac{1}{2}$ and $2\mu_1+3\mu'\geq1$ (where, by \eqref{eq:mu1}, 
$\mu_1 =  \sum_{i=0}^{3}(u_i-v_i)$
and 
$\mu'=\sum_{i=0}^{3}(\frac{2i}{3}-1)(u_i+v_i)$.)
\begin{eqnarray*}
h_{\mathcal{D}}(p) &=& 
    u_0 \left[p(1 - p^3) + (1-p)^4 \right]\\
&&+ u_1 \left[p(1 - (1-p)p^2) + (1-p)^3 p\right]\\
&&+ u_2 \left[p(1 - (1-p)^2p) + (1-p)^2 p^2\right]\\
&&+ u_3 \left[p(1 - (1-p)^3) + (1-p) p^3\right]\\
&& +v_0 \left[(1-p)(1 - p^3) + (1-p)^3p\right]\\
&& +v_1 \left[(1-p)(1 - (1-p)p^2) + (1-p)^2p^2\right]\\
&& +v_2 \left[(1-p)(1 - (1-p)^2p) + (1-p)p^3\right]\\
&& +v_3 \left[(1-p)(1 - (1-p)^3) + p^4\right]\\
&=& u_0 + v_0 + v_1 + v_2\\ 
&& + p \cdot (-3u_0 + 2u_1 +u_2 - v_1 - 2v_2 + 3v_3)\\
&& + p^2 \cdot (6u_0-3u_1+3u_3-3v_0+3v_2 -6v_3)\\
&& + p^3 \cdot (-4u_0+2u_1-2u_3 +2v_0 - 2v_2+4v_3 )
\end{eqnarray*}

Every distribution (whose marginals are) in $\widetilde{K}^N_{1/2}(\mon_4)$ must satisfy the following system of equations and inequalities,
where \eqref{eqn:hd1}--\eqref{eqn:hd4} are equivalent to $h_{\mathcal{D}}(p) = \frac{1}{2}$, and \eqref{eqn:validDist1}--\eqref{eqn:validDist3} guarantee that $u_i$s and $v_i$s describe a distribution.

\begin{align}
&u_0 + v_0 + v_1 + v_2 = \frac{1}{2}        \label{eqn:hd1}\\
&-3u_0 + 2u_1 +u_2 - v_1 - 2v_2 + 3v_3 = 0  \label{eqn:hd2}\\
&6u_0-3u_1+3u_3-3v_0+3v_2 -6v_3 = 0\label{eqn:hd3}\\
&-4u_0+2u_1-2u_3 +2v_0 - 2v_2+4v_3 = 0  \label{eqn:hd4}\\
&\sum_{i=0}^3 (u_i+v_i) = 1\label{eqn:validDist1}\\
&u_i \geq 0, \quad \forall 0 \leq i \leq 3\label{eqn:validDist2}\\
&v_i \geq 0, \quad \forall 0 \leq i \leq 3\label{eqn:validDist3}
\end{align}

Summing up \eqref{eqn:hd2} multiplied by $3$, \eqref{eqn:hd4} multiplied by $-13/6$, and \eqref{eqn:validDist1} multiplied by $2/3$, we have that
\begin{align*}
2/3 
&= u_0/3+7u_1/3+11u_2/3+5u_3-11v_0/3-7v_1/3-v_2+v_3\\
&\geq -u_0 + u_1 + 3u_2 + 5u_3 - 5v_0 - 3v_1 - v_2 + v_3\\
&= 2\mu_1+3\mu' \;,
\end{align*}
where the last equality uses \eqref{eq:mu1}. By \cref{lem:ky1}, $\widetilde{K}^Y_1({\mon_4})=\{(\mu_1,\mu')\in[-1,1]^2\colon 2\mu_1+3\mu'\geq1\}$, and from the above inequality every vector $(\mu_1,\mu')\in \widetilde{K}^N_{1/2}({\mon_4})$ satisfies $2\mu_1+3\mu'\leq2/3$. This implies that $\widetilde{K}^Y_1({\mon_4})\cap\widetilde{K}^N_{1/2}({\mon_4})=\emptyset$, and finishes the proof.
\end{proof}

\subsection{Balanced LTFs on 4 variables}
In this section, we prove \cref{thm:4ltfs}.
\thmthree*
We remark that there are non-balanced LTFs on fewer than four variables that are approximation resistant. For example, if $f(x_1,x_2)=x_1 \textsf{ OR } x_2$, then $\maxf$ is approximation resistant to space $o(n)$ even in the larger class of streaming algorithms (see, e.g., Corollary~4.2 in~\cite{CGV20}).
\begin{proof}[Proof of Theorem~\ref{thm:4ltfs}]
After relabeling and negating some of the variables of $f$, we can assume that $f(x_1,x_2,x_3,x_4)=\sign(w_1 x_1+w_2 x_2+w_3 x_3+w_4 x_4)$, where $w_1\geq w_2\geq w_3\geq w_4\geq 0$ (if $f$ depends on $i<4$ variables, then we set $w_{i+1}=\ldots=w_4=0$).
Since $f$ is balanced, $\xi_1 w_1 + \xi_2 w_2 +\xi_3 w_3 +\xi_4 w_4 \neq 0$ for all $\xi_i\in\{-1,1\}$. Now consider the following three cases.
\begin{itemize}
    \item If $w_1>w_2+w_3+w_4$, then $f=\sign(x_1)$ is a dictator function, so $\maxf$ can be trivially $(1-\eps)$-approximated in $O(\log(n)/\eps^2)$ space by an $\ell_1$-sketch algorithm~\cite{Indyk,KNW10}.
    \item If $w_2+w_3-w_4<w_1<w_2+w_3+w_4$, then $f=\mon_4$ is a monarchy function on $k=4$ variables. Indeed, in this case only the sum of the votes of the three last variables overrules the vote of the first variable. By \cref{lem:mon4}, $\maxf$ is sketching approximable in $O(\log(n))$ space.
    \item If $w_1<w_2+w_3-w_4$, then $f=\maj(x_1,x_2,x_3)$ is the majority function on $3$ variables. Indeed, the sum of any two weights of the first three variables outweighs the sum of the remaining weights. In this case, $\maxf$ is known to be sketching approximable in space $O(\log(n))$ (this follows from the characterization of sketching approximable symmetric functions in \cite[Lemma~2.14]{CGSV21-boolean} and the fact that a balanced LTF doesn't support one-wise independent distributions). 
    
    Another way to see that the majority function is sketching approximable is via \cref{thm:approximablefn}. Indeed, since majority is a symmetric function, the (non-empty) Chow parameters of the majority function are all equal and non-zero (see, e.g., \cite[Theorem~5.19]{ryansbook} for the exact values of the Fourier coefficients of the majority function). Then the Chow parameters define the majority function itself, and, by \cref{thm:approximablefn}, $\maxf$ is sketching approximable in space $O(\log(n))$. \qedhere
\end{itemize}
\end{proof}

\section{Approximation resistance of Monarchy Functions}\label{sec:proof of main thm}
In this section, we prove \cref{thm:main}: we show that for $k\geq 5$, the $\mon_k$ function is approximation resistant. Recall that by~\cref{lem:intersection} it suffices to show that $\widetilde{K}^Y_1({\mon_k})\cap \widetilde{K}^N_{1/2}(\mon_k)\neq\emptyset$ for $k\geq5$. 

In the following we show that for $k\geq5$, there exist vectors $(\vecu,\vecv)$ with certain properties that will be useful in showing that $\widetilde{K}^Y_1({\mon_k})\cap \widetilde{K}^N_{1/2}(\mon_k)\neq\emptyset$.

\begin{lemma}\label{lem:main sufficient}
For every $k\geq5$, there exists $\vecu,\vecv\in\R^k_{\geq0}$ satisfying the following conditions.
\begin{enumerate}[(i)]
\item $\sum_i (u_i+v_i)=1$, i.e., $\vecu,\vecv$ define a distribution $\cD$.
In particular, the marginals of $\cD$ is $(\mu_1,\mu',\dots,\mu')$ where $\mu_1=\sum_i(u_i-v_i)$, and $\mu'=\sum_i(\frac{2i}{k-1}-1)(u_i+v_i)$.
\item $\vecu$ and $\vecv$ satisfy 
\begin{align*}
&(1/2-\delta)\sum_{i=0}^{k-1}u_i+(1/2+\delta)\sum_{i=0}^{k-1}v_i\\
+&\sum_{i=0}^{k-1}u_i\left(- (1/2+\delta)^i(1/2-\delta)^{k-i} + (1/2-\delta)^i(1/2+\delta)^{k-i}\right)\nonumber\\
+&\sum_{i=0}^{k-1}v_i\left(- (1/2+\delta)^{i+1}(1/2-\delta)^{k-1-i}+(1/2-\delta)^{i+1}(1/2+\delta)^{k-1-i}\right)\nonumber\\
=&1/2
\end{align*}
for every $\delta\in[-1/2,1/2]$. In particular, this implies that $\cD\in S^N_{1/2}$.
\item $p'\geq1-\frac{k-2}{k-1}p_1$ where $p'=\Pr_{\vecx\sim\cD}[x_2=1]=\frac{1}{k-1}\left(\sum_iiu_i+\sum_iiv_i\right)$ and $p_1=\Pr_{\vecx\sim\cD}[x_1=1]=\sum_iu_i$.
In particular, this implies the existence of $\cD_Y\in S^Y_1$ and $\vecmu(\cD_Y)=(\mu_1,\mu',\dots,\mu')$.
\end{enumerate}
\end{lemma}

Now, we are ready to prove~\cref{thm:main} using~\cref{lem:main sufficient} and~\cref{thm:cgsv}.

\thmone*
\begin{proof}
For every $k\geq5$, let $\vecu,\vecv\in\R^k_{\geq0}$, and $\mu_1,\mu'\in[-1,1]$ be the vectors given by~\cref{lem:main sufficient}. Note that condition (i) guarantees that $\vecu,\vecv$ define a distribution $\cD$ with marginal $(\mu_1,\mu',\dots,\mu')$.

First, we show that condition (ii) is a sufficient condition for $(\mu_1,\mu')\in \widetilde{K}^N_{1/2}$. Recall that $\cD_N\in S_{1/2}^N(\mon_k)$ if for every $\delta\in[-1/2,1/2]$, $\E_{\vecb\in\cD_N}\E_{\veca\sim\bern(1/2+\delta)}[\mon_k(\vecb\odot \veca)]=1/2$. Since $\Pr_\vecx[\mon_k(\vecx)=1]=\Pr_\vecx[x_1=1]-\Pr_\vecx[\vecx=10^{k-1}]+\Pr_\vecx[\vecx=01^{k-1}]$, we have that
\[
    \E_{\vecb\in\cD_N}\E_{\veca\sim\bern(1/2+\delta)}[\mon_k(\vecb\odot \veca)]=\Pr_{\vecb,\veca}[\vecb_1\odot\veca_1=1]-\Pr_{\vecb,\veca}[\vecb\odot\veca=1(-1)^{k-1}]+\Pr_{\vecb,\veca}[\vecb\odot\veca=(-1)1^{k-1}] \,.
\]
We compute these three probabilities in terms of $\vecu,\vecv,\delta$. 
\begin{align*}
   \Pr_{\vecb,\veca}[\vecb_1\odot\veca_1=1]&=(1/2-\delta)\sum_{i=0}^{k-1}u_i+(1/2+\delta)\sum_{i=0}^{k-1}v_i\;,\\
    \Pr_{\vecb,\veca}[\vecb\odot\veca=1(-1)^{k-1}]
    &=\sum_{i=0}^{k-1}u_i(1/2+\delta)^i(1/2-\delta)^{k-i}+\sum_{i=0}^{k-1}v_i(1/2+\delta)^{i+1}(1/2-\delta)^{k-1-i} \;,\\
    \Pr_{\vecb,\veca}[\vecb\odot\veca=(-1)1^{k-1}]
    &=\sum_{i=0}^{k-1}u_i(1/2-\delta)^i(1/2+\delta)^{k-i}+\sum_{i=0}^{k-1}v_i(1/2-\delta)^{i+1}(1/2+\delta)^{k-1-i}  \;.\\
\end{align*}
Note that condition (ii) implies that 
\[
\Pr_{\vecb,\veca}[\vecb_1\odot\veca_1=1] + \Pr_{\vecb,\veca}[\vecb\odot\veca=1(-1)^{k-1}] + \Pr_{\vecb,\veca}[\vecb\odot\veca=(-1)1^{k-1}] = \frac{1}{2}
\]
for every $\delta\in[-1/2,1/2]$ as desired. This implies that $\cD\in S_{1/2}^N(\mon_k)$. As condition (i) gives $\vecmu(\cD_N)=(\mu_1,\mu',\ldots,\mu')$, we have $(\mu_1,\mu')\in \widetilde{K}^N_{1/2}$ as desired.

Next, as $p'=\frac{\mu'+1}{2}$ and $1-\frac{k-2}{k-1}p_1=1-\frac{(k-2)(\mu_1+1)}{2(k-1)}$, condition (iii) implies $\mu_1(k-2)+\mu'(k-1)\geq1$. By~\cref{lem:ky1}, this implies that $(\mu_1,\mu')\in \widetilde{K}_1^Y(\mon_k)$ as desired.

To sum up,~\cref{lem:main sufficient} gives us
$(\mu_1,\mu')\in \widetilde{K}^Y_1\cap \widetilde{K}^N_{1/2}$ for every $k\geq5$
and~\cref{lem:intersection} implies $(\mu_1,\mu',\ldots,\mu')\in K^Y_1\cap K^N_{1/2}$. By~\cref{thm:cgsv}, we conclude that $\mon_k$ is sketching approximation resistant to space $o(\sqrt{n})$ and, hence, complete the proof of~\cref{thm:main}.
\end{proof}

\subsection{Proof of Lemma~\ref{lem:main sufficient}}
In the proof of \cref{lem:main sufficient} we will use the following combinatorial identity.
\begin{lemma}\label{lem:comb}
For every $\delta\in[-1/2,1/2]$ and $m\in\N$,
\begin{align*}
&\;\;\;\sum_{i=\ceil{m/2}}^{m}(1/2+\delta)^{i+1}(1/2-\delta)^{m-i}\left(\binom{m}{i}-\binom{m}{i+1}\right)\\
&-\sum_{i=\ceil{m/2}}^{m}(1/2-\delta)^{i+1}(1/2+\delta)^{m-i}\left(\binom{m}{i}-\binom{m}{i+1}\right)\\
&=2\delta \,.
\end{align*}
\end{lemma}
\begin{proof}
Let $X_1,\ldots,X_{m+1}$ be independent identically distributed random variables, each having the distribution $\bern(1/2+\delta)$. For $j\in\{0,\ldots,m+1\}$, let $\one_j$ be the indicator of the event that exactly $j$ variables from $X_1,\ldots,X_{m+1}$ are ones.
First observe that for $i\in\{0,\ldots,m+1\}$,
\[
\E[x_1\cdot \one_i] = (1/2+\delta)^{i}(1/2-\delta)^{m+1-i}\left(\binom{m}{i-1}-\binom{m}{i}\right) \,.
\]
Using the above, we are going to show that the left hand side of the equation in~\cref{lem:main sufficient} equals to $\sum_{i=0}^{m+1}\E[x_1\cdot\one_i]=\E[x_1]=2\delta$. By changing summations' limits and updating the binomial coefficients accordingly, we have
\begin{align*}
    &\;\;\;\sum_{i=\ceil{m/2}}^{m}(1/2+\delta)^{i+1}(1/2-\delta)^{m-i}\left(\binom{m}{i}-\binom{m}{i+1}\right)\\
&-\sum_{i=\ceil{m/2}}^{m}(1/2-\delta)^{i+1}(1/2+\delta)^{m-i}\left(\binom{m}{i}-\binom{m}{i+1}\right) \, .\\
&=\sum_{i=0}^{\floor{m/2}}(1/2+\delta)^{m-i+1}(1/2-\delta)^{i}\left(\binom{m}{i}-\binom{m}{i-1}\right)\\
&-\sum_{i=\ceil{m/2}+1}^{m+1}(1/2-\delta)^{i}(1/2+\delta)^{m-i+1}\left(\binom{m}{i-1}-\binom{m}{i}\right) \, .\\
\intertext{Using $\binom{m}{\floor{m/2}}=\binom{m}{\ceil{m/2}}$, we update the first summation's limits:}
&=\sum_{i=0}^{\ceil{m/2}}(1/2+\delta)^{m-i+1}(1/2-\delta)^{i}\left(\binom{m}{i}-\binom{m}{i-1}\right)\\
&-\sum_{i=\ceil{m/2}+1}^{m+1}(1/2-\delta)^{i}(1/2+\delta)^{m-i+1}\left(\binom{m}{i-1}-\binom{m}{i}\right)\\
&=\sum_{i=0}^{m}(1/2-\delta)^{i}(1/2+\delta)^{m-i+1}\left(\binom{m}{i}-\binom{m}{i-1}\right)\\
 &=\sum_{i=0}^{m+1} \E[x_1\cdot \one_i] 
 = \E[x_1] 
 = 2\delta \,,
\end{align*}
which concludes the proof.
\end{proof}

We are now ready to prove~\cref{lem:main sufficient}.
\begin{proof}[Proof of~\cref{lem:main sufficient}]
We prove this lemma by considering three cases: $k=5$, $k>5$ is even, and $k>5$ is odd.
\paragraph{Case I: $k=5$.}
In this case, we consider the following pair of vectors
\begin{align*}
    \vecu = (u_0, u_1, u_2, u_3, u_4)=&\left(0,0, 0,  0, \frac{1}{3}\right) \;,\\
    \vecv = (v_0, v_1, v_2, v_3, v_4)=&\left(0, 0, \frac{1}{3}, \frac{1}{6}, \frac{1}{6}\right) \;.
\end{align*} 
\begin{enumerate}[(i)]
\item It's straightforward to verify that $\sum_i(u_i+v_i)=1$.
\item For $\delta\in[-1/2,1/2)$, using the substitution $y=(1/2+\delta)/(1/2-\delta)$, we have
    \begin{align*}
         &(1/2-\delta)\sum_{i=0}^{k-1}u_i+(1/2+\delta)\sum_{i=0}^{k-1}v_i\\
         +&\sum_{i=0}^{k-1}u_i\left(- (1/2+\delta)^i(1/2-\delta)^{k-i} + (1/2-\delta)^i(1/2+\delta)^{k-i}\right)\\
         +&\sum_{i=0}^{k-1}v_i\left(- (1/2+\delta)^{i+1}(1/2-\delta)^{k-1-i}+(1/2-\delta)^{i+1}(1/2+\delta)^{k-1-i}\right)\\
         =& 1/2+\delta/3\\
         +&1/3(-(1/2+\delta)^4(1/2-\delta)+(1/2-\delta)^4(1/2+\delta))\\
         +&1/3(-(1/2+\delta)^3(1/2-\delta)^2+(1/2-\delta)^3(1/2+\delta)^2)\\
         +&1/6(-(1/2+\delta)^4(1/2-\delta)+(1/2-\delta)^4(1/2+\delta))\\
         +&1/6(-(1/2+\delta)^5+(1/2-\delta)^5)\\
         =&1/2+\delta/3+(1/2-\delta)^5(-y^4/3+y/3-y^3/3+y^2/3-y^4/6+y/6-y^5/6+1/6)\\
         =&1/2+\delta/3-(1/2-\delta)^5(y-1)(y+1)^4/6\\
         =&1/2+\delta/3-(1/2-\delta)^5\left(\frac{2\delta}{1/2-\delta}\right)\left(\frac{1}{1/2-\delta}\right)^4/6\\
         =&1/2+\delta/3-2\delta/6=1/2\;.
    \end{align*}
    For $\delta=1/2$, it's easy to see that the sum above equals~$1/2$, too.
\item Since $p_1 = \sum_iu_i=1/3$ and $p'=\frac{1}{k-1}\left(\sum_iiu_i+\sum_iiv_i\right)=19/24$, the inequality $p'\geq 1-\frac{k-2}{k-1}p_1$ holds.

\end{enumerate}
\paragraph{Case II: $k>5$ is even.}
Let $T=\binom{k}{k/2}-2$. Consider the vectors $\vecu,\vecv\in\R_{\geq0}^k$ as follows.
\begin{align*}
    u_i=\left\{\begin{array}{ll}
\frac{T-2}{2T}    & ,\  \text{if }i=k/2\\
0    & ,\ \text{otherwise.}
\end{array}
\right.
~~~ \text{and} ~~~
 v_i=\left\{\begin{array}{ll}
\frac{\binom{k-1}{i}-\binom{k-1}{i+1}}{T}    & ,\  \text{if } i\geq k/2\\
0    & ,\ \text{otherwise.}
\end{array}
\right.
\end{align*}
\begin{enumerate}[(i)]
    \item Note that
    \[\sum_{i=0}^{k-1}v_i=\frac{1}{T}\sum_{i=k/2}^k\left(\binom{k-1}{i}-\binom{k-1}{i+1}\right)=\frac{1}{T}\binom{k-1}{k/2}=\frac{1}{2T}\binom{k}{k/2}=\frac{T+2}{2T} \,.
    \]
    Thus,
    \[
    \sum_{i=0}^{k-1}u_i+\sum_{i=0}^{k-1}v_i
    =\frac{T-2}{2T}+\frac{T+2}{2T}=1 \,.
    \]
    \item 
    From the definition of $\vecu$ and $\vecv$, using $\sum_{i=0}^{k-1}v_i=\frac{T+2}{2T}$ and applying \cref{lem:comb} with $m=k-1$, we have that for every $\delta\in[-1/2,1/2]$,
    \begin{align*}
    &(1/2-\delta)\sum_{i=0}^{k-1}u_i+(1/2+\delta)\sum_{i=0}^{k-1}v_i\\
    +&\sum_{i=0}^{k-1}u_i\left(- (1/2+\delta)^i(1/2-\delta)^{k-i} + (1/2-\delta)^i(1/2+\delta)^{k-i}\right)\\
    +&\sum_{i=0}^{k-1}v_i\left(- (1/2+\delta)^{i+1}(1/2-\delta)^{k-1-i}+(1/2-\delta)^{i+1}(1/2+\delta)^{k-1-i}\right)\\
    =&\left(\frac{T-2}{2T}\left(1/2-\delta
    \right)+\frac{T+2}{2T}(1/2+\delta)\right)+0\\
    +&\frac{1}{T}\sum_{i=k/2}^{k-1}\left(\binom{k-1}{i}-\binom{k-1}{i+1}\right)\left(-(1/2+\delta)^{i+1}(1/2-\delta)^{k-1-i}+(1/2-\delta)^{i+1}(1/2+\delta)^{k-1-i}\right)\\
    =&\left(1/2+2\delta/T\right)-2\delta/T=1/2\;.
    \end{align*}
    \item From the definition of $\vecu$ and $\vecv$, we have that $p_1 = \sum_{0}^{k-1}u_i=\frac{T-2}{2T}$.
    \begin{align*}
        p' &= \frac{1}{k-1}\left(\sum_{0}^{k-1}iu_i+\sum_{0}^{k-1}iv_i\right)\\
        &=\frac{1}{T(k-1)}\left((T-2)k/4+\sum_{k/2}^{k-1}i\left(\binom{k-1}{i}-\binom{k-1}{i+1}\right)\right)\\
        &=\frac{1}{T(k-1)}\left((T-2)k/4+(k/2-1)\binom{k-1}{k/2}+\sum_{i=k/2}^{k-1}i\binom{k-1}{i} -\sum_{i=k/2-1}^{k-1}i\binom{k-1}{i+1}\right)\\
        &=\frac{1}{T(k-1)}\left((T-2)k/4+(k/2-1)\binom{k-1}{k/2}+\sum_{i=k/2}^{k-1}i\binom{k-1}{i} -\sum_{i=k/2}^{k}(i-1)\binom{k-1}{i}\right)\\
        &=\frac{1}{T(k-1)}\left((T-2)k/4+(k/2-1)\binom{k-1}{k/2}+\sum_{i=k/2}^{k-1}\binom{k-1}{i}\right)\\
        &=\frac{1}{T(k-1)}\left((T-2)k/4+(k/2-1)\binom{k-1}{k/2}+2^{k-2}\right) \, .\\
        \intertext{Using $\binom{k-1}{k/2}=\frac{1}{2}\binom{k}{k/2}=(T+2)/2$}
        &=\frac{1}{T(k-1)}\left((T-2)k/4+(k/2-1)(T+2)/2+2^{k-2}\right) \, .\\
        \intertext{Using $2^{k-2}\geq k+(\binom{k}{k/2}-2)/2=k+T/2$ for $k\geq6$}
        &\geq\frac{Tk/2+k-1}{T(k-1)}\\
        &=1-\frac{Tk-2T-2k+2}{2T(k-1)}\\
        &>1-\frac{k-2}{k-1}\cdot\frac{T-2}{2T}\\
        &=1-\frac{k-2}{k-1}p_1\,.
    \end{align*}
\end{enumerate}
\paragraph{Case III: $k>5$ is odd.}
Let $T=2\binom{k-1}{\frac{k-1}{2}}-2$. Consider the vectors $\vecu,\vecv\in\R_{\geq0}^k$ as follows.
\begin{align*}
    u_i=\left\{\begin{array}{ll}
\frac{T-2}{4T}    & ,\  \text{if } i=\frac{k-1}{2}\text{ or }i=\frac{k+1}{2}\\
0    & ,\ \text{otherwise.}
\end{array}
\right.
~~~ \text{and} ~~~
 v_i=\left\{\begin{array}{ll}
\frac{\binom{k-1}{i}-\binom{k-1}{i+1}}{T}    & ,\  \text{if } i\geq \frac{k-1}{2}\\
0    & ,\ \text{otherwise.}
\end{array}
\right.
\end{align*}
\begin{enumerate}[(i)]
    \item Similarly to Case~II, $\sum_{i=0}^{k-1}v_i=\frac{T+2}{2T}$ and $
    \sum_{i=0}^{k-1}u_i+\sum_{i=0}^{k-1}v_i=1$.
    \item Using $\sum_{i=0}^{k-1}v_i=\frac{T+2}{2T}$ and \cref{lem:comb} with $m=k-1$, we conclude that for every $\delta\in[-1/2,1/2]$,
    \begin{align*}
    &(1/2-\delta)\sum_{i=0}^{k-1}u_i+(1/2+\delta)\sum_{i=0}^{k-1}v_i\\
    +&\sum_{i=0}^{k-1}u_i\left(-(1/2+\delta)^i(1/2-\delta)^{k-i} + (1/2-\delta)^i(1/2+\delta)^{k-i}\right)\\
    +&\sum_{i=0}^{k-1}v_i\left(- (1/2+\delta)^{i+1}(1/2-\delta)^{k-1-i}+(1/2-\delta)^{i+1}(1/2+\delta)^{k-1-i}\right)\\
    =&\left(\frac{2(T-2)}{4T}(1/2-\delta)
    +\frac{T+2}{2T}(1/2+\delta)\right)\\
    +&\frac{T-2}{4T}\left(-(1/2+\delta)^{\frac{k-1}{2}}(1/2-\delta)^{\frac{k+1}{2}} + (1/2-\delta)^{\frac{k-1}{2}}(1/2+\delta)^{\frac{k+1}{2}}\right)\\
    +&\frac{T-2}{4T}\left(-(1/2+\delta)^{\frac{k+1}{2}}(1/2-\delta)^{\frac{k-1}{2}} + (1/2-\delta)^{\frac{k+1}{2}}(1/2+\delta)^{\frac{k-1}{2}}\right)\\
    +&\frac{1}{T}\sum_{i=(k-1)/2}^{k-1}\left(\binom{k-1}{i}-\binom{k-1}{i+1}\right)\left(-(1/2+\delta)^{i+1}(1/2-\delta)^{k-1-i}+(1/2-\delta)^{i+1}(1/2+\delta)^{k-1-i}\right)\\
    =&\left(1/2+2\delta/T\right)
    -2\delta/T=1/2\;.
    \end{align*}
    \item Similarly to the previous case, $p_1 = \sum_{0}^{k-1}u_i=2\frac{T-2}{4T}=\frac{T-2}{2T}$, and
    \begin{align*}
        p' &= \frac{1}{k-1}\left(\sum_{i=0}^{k-1}iu_i+\sum_{i=0}^{k-1}iv_i\right)\\
        &=\frac{1}{T(k-1)}\left(\frac{(T-2)}{4}\times\left(\frac{k - 1}{2} + \frac{k + 1}{2}\right) +\sum_{i=(k-1)/2}^{k-1}i\left(\binom{k-1}{i}-\binom{k-1}{i+1}\right)\right)\\
        &=\frac{1}{T(k-1)}\left((T-2)k/4+\left(\frac{k-1}{2}-1\right)\binom{k-1}{(k-1)/2}\right.\\
        &\left.\;\;\;\;\;\;\;\;\;\;\;\;\;\;\;\;\;\;\;\;\;\;+\sum_{i=(k-1)/2}^{k-1}i\binom{k-1}{i} -\sum_{i=(k-1)/2-1}^{k-1}i\binom{k-1}{i+1}\right)\\
        &=\frac{1}{T(k-1)}\left((T-2)k/4+\left(\frac{k-1}{2}-1\right)\binom{k-1}{(k-1)/2}\right.\\
        &\left.\;\;\;\;\;\;\;\;\;\;\;\;\;\;\;\;\;\;\;\;\;\;+\sum_{i=(k-1)/2}^{k-1}i\binom{k-1}{i} -\sum_{i=(k-1)/2}^{k}(i-1)\binom{k-1}{i}\right)\\
        &=\frac{1}{T(k-1)}\left((T-2)k/4+\left(\frac{k-1}{2}-1\right)\binom{k-1}{(k-1)/2}+\sum_{i=(k-1)/2}^{k-1}\binom{k-1}{i}\right)\\
        &=\frac{1}{T(k-1)}\left((T-2)k/4+\left(\frac{k-1}{2}-1\right)\binom{k-1}{(k-1)/2}+\binom{k-1}{\leq(k-1)/2}\right) \, .\\
        \intertext{Using $\binom{k-1}{(k-1)/2}=(T+2)/2$}
        &=\frac{1}{T(k-1)}\left((T-2)k/4+\left(\frac{k-1}{2}-1\right)(T+2)/2+\binom{k-1}{\leq(k-1)/2}\right)\\
        &=\frac{1}{T(k-1)}\left(Tk/2-3T/4-3/2
        +\binom{k-1}{\leq(k-1)/2}\right)\\
        \intertext{Using $\binom{k-1}{\leq(k-1)/2}\geq \frac{3}{2}\binom{k-1}{\frac{k-1}{2}}+k-2=3T/4+k-1/2$ which holds for every $k\geq 7$}
        &\geq\frac{Tk/2+k-2}{T(k-1)}\\
        &=1-\frac{Tk-2T-2k+4}{2T(k-1)}\\
        &=1-\frac{k-2}{k-1}\cdot\frac{T-2}{2T}\\
        &=1-\frac{k-2}{k-1}p_1\,.
    \end{align*}
\end{enumerate}
This concludes the proof of~\cref{lem:main sufficient}.
\end{proof}
    
\section{Chow parameters and the approximability of weak monarchies}\label{sec:Theorem2_proof}

In this section, we prove that infinitely many weak monarchy functions are sketching approximable within $O(\log(n))$ space. We first prove in \cref{sec:chowparam,sec:lemma_proofs} that every LTF defined by its Chow parameters (i.e., degree-$1$ Fourier coefficients as weights and threshold $0$) is sketching approximable within $O(\log(n))$ space. And later in \cref{sec:weakmonarchy}, we prove that infinitely many weak monarchy functions are balanced LTFs defined by their Chow parameters. 

\subsection{Approximability of LTFs defined by their Chow parameters}\label{sec:chowparam}

\begin{theorem}\label{lem:approximablefn}
For every Boolean function $f:\{-1,1\}^k\rightarrow \{0,1\}$ of the form \[f(x) = \sign \left(\sum_{i=1}^k \widehat{f}(\{i\})x_i \right)\;,\] $\maxf$ is sketching approximable in $O(\log(n))$ space.
\end{theorem}

\begin{definition}\label{def:epsilon}
Define $\epsilon_0(f) = \min\{\sum_{i=1}^k \widehat{f}(\{i\})\cdot x_i : f(x)=1\}$. Define $\epsilon^*(f) = \min\{\frac{\epsilon_0(f)}{3k},\frac{2\epsilon_0(f)^2}{9\rho(f)k^2}\}$.
\end{definition}

We will use the following theorem to prove \autoref{lem:approximablefn}.

\begin{theorem}\label{thm:approximablefn}
For every Boolean function $f:\{-1,1\}^k\rightarrow \{0,1\}$ and every $\epsilon>0$, there exists an $O(\log(n))$ space $(\rho(f)+\epsilon^*(f)-\epsilon)$-approximation algorithm for $\mcsp(f)$. \end{theorem}

First we show how to prove \cref{lem:approximablefn} using \cref{thm:approximablefn}.
\begin{proof}[Proof of~\cref{lem:approximablefn}]
If $f(x)$ is the constant zero function, then it's trivially approximable in $O(\log(n))$ space. Otherwise, when $f(x) = \sign \left(\sum_{i=1}^k \widehat{f}(\{i\})\cdot x_i \right)$, we have $\epsilon_0(f) = \min\{\sum_{i=1}^k \widehat{f}(\{i\})\cdot x_i : f(x)=1\}>0$ and hence $\epsilon^*(f)>0$ by their definitions. Now for $\eps=\epsilon^*(f)/2$, \cref{thm:approximablefn} implies that there is a $(\rho(f)+\epsilon^*(f)/2)$-approximation algorithm for $\mcsp(f)$, and finishes the proof.
\end{proof}

Before we prove \cref{thm:approximablefn}, we will describe some useful definitions and lemmas from \cite{CGSV21-boolean}. 

Let $f:\{-1,1\}^k\rightarrow\{0,1\}$ be a Boolean constraint function of arity $k$ and $X_1,\dots,X_n$ be variables. A constraint $C$ consists of $\vecj=(j_1,\dots,j_k)\in[n]^k$ and $\vecb=(b_1,\dots,b_k)\in\{-1,1\}^k$ where the $j_i$'s are distinct. The constraint $C$ reads as requiring $f(\vecb\odot\veccx|_\vecj)=f(b_1X_{j_1},\dots,b_kX_{j_k})=1$. A \textsf{Max-CSP}($f$) instance $\Psi$ contains $m$ constraints $C_1,\dots,C_m$ with non-negative weights $w_1,\ldots,w_m$ where $C_i=(\vecj(i),\vecb(i))$ and $w_i \in \R$ for each $i\in[m]$. For an assignment $\vecsigma\in\{-1,1\}^n$, the value $\val_\Psi(\vecsigma)$ of $\vecsigma$ on $\Psi$ is the fraction of weight of constraints satisfied by $\vecsigma$, i.e., $\val_\Psi(\vecsigma)=\tfrac{1}{W}\sum_{i\in[m]}w_i \cdot f(\vecb(i)\odot\vecsigma|_{\vecj(i)})$, where $W = \sum_{i=1}^m w_i$. The optimal value of $\Psi$ is defined as $\val_\Psi=\max_{\vecsigma\in\{-1,1\}^n}\val_\Psi(\vecsigma)$.

\begin{definition}[Bias (vector)]
For $\veclambda = (\lambda_1,\ldots,\lambda_k) \in \R^k$, and instance $\Psi = (C_1,\ldots,C_m; w_1,\ldots,w_m)$ of $\maxf$ where $C_i = (\vecj(i),\vecb(i))$ and $w_i \geq 0$, we let the {\em $\veclambda$-bias vector} of $\Psi$, denoted $\bias_{\veclambda}(\Psi)$, be the vector in $\R^n$ given by 
\[
\bias_{\veclambda}(\Psi)_\ell = \frac{1}W \cdot \sum_{i \in [m], t \in [k] : j(i)_t = \ell} \lambda_t w_i \cdot b(i)_t \, ,
\]
for $\ell \in [n]$, where $W = \sum_{i \in [m]} w_i$. 
The $\veclambda$-bias of $\Psi$, denoted $B_{\veclambda}(\Psi)$, is the $\ell_1$ norm of $\bias_{\veclambda}(\Psi)$, i.e., $B_{\veclambda}(\Psi) = \sum_{\ell=1}^n |\bias_{\veclambda}(\Psi)_\ell|$.
\end{definition}

\begin{lemma}[{\cite[Lemma~4.7]{CGSV21-boolean}}]\label{lem:bias}
For every $\veclambda \in \mathbb{R}^k$, we have $B_{\veclambda}(\Psi) = \max_{a\in \{-1,1\}^n} \langle a,  \bias_{\veclambda}(\Psi)\rangle$.
\end{lemma}

\begin{lemma}[{\cite[Lemma~4.4]{CGSV21-boolean}}]\label{prop:ell1norm}
For every vector $\veclambda \in \R^k$ and $\epsilon > 0$, there exists a $O(\log(n))$ 
space sketching algorithm $\cA$ that on input a stream $\sigma_1,\ldots,\sigma_\ell$, representing an instance $\Psi = (C_1,\ldots,C_m;w_1,\ldots,w_m)$, outputs a $(1\pm \epsilon)$-approximation to $B_{\veclambda}(\Psi)$, i.e., for every $\Psi$, $(1-\epsilon)B_{\veclambda}(\Psi) \leq \cA(\Psi) \leq (1+\epsilon)B_{\veclambda}(\Psi)$, with probability at least $2/3$. 
\end{lemma}

Below, we describe \cref{alg:Bias} and show that it is an $O(\log(n))$ space $(\rho(f)+\epsilon^*(f)-\epsilon)$-approximation algorithm for $\mcsp(f)$.
\begin{algorithm}[H]
	\caption{A sketching $(\rho(f)+\epsilon^*(f)-\epsilon)$-approximation algorithm for $\mcsp(f)$}
	\label{alg:Bias}
\begin{algorithmic}[1]
		\Input a stream $\sigma_1,\ldots,\sigma_\ell$ representing an instance $\Psi$ of $\maxf$ where $\sigma_i = ((\vecj(i),\vecb(i)),w_i)$.
		    \State  Let $\veclambda = (\widehat{f}(\{1\}),\dots,\widehat{f}(\{k\}))\in \mathbb{R}^k $ and $\epsilon' = \epsilon/8$.
		    \State Use the algorithm $\cA$ from \cref{prop:ell1norm} to  compute $\tilde{B}$ to be a $(1\pm\epsilon')$ approximation to $B_{\veclambda}(\Psi)$, i.e., $(1-\epsilon')B_{\veclambda}(\Psi) \leq \tilde{B} \leq (1+\epsilon')B_{\veclambda}(\Psi)$ with probability at least $2/3$.
		    \State Let $\tilde{\delta} = \min\{\frac{1}{3k},\frac{2\tilde{B}}{9\rho(f)k^2}\}$.
		    \State {\bf Output:} $v= \rho(f) + \frac{\tilde{B}\tilde{\delta}}{(1+\epsilon')^2}$.
	\end{algorithmic}
\end{algorithm}
It is clear that the algorithm above runs in $O(\log(n))$ space (in particular by \cref{prop:ell1norm} for Step 2). We now turn to analyzing the correctness of the algorithm.

\subsubsection{Analysis of the correctness of Algorithm~\ref{alg:Bias}}
Before we analyse \cref{alg:Bias}, we establish some upper and lower bounds on $\val_{\Psi}$ in terms of $B_{\veclambda}(\Psi)$ where $\veclambda=(\widehat{f}(\{1\}),\dots,\widehat{f}(\{k\}))$.

\begin{lemma}[Lower bound on $\val_{\Psi}$]\label{lemma:lb}
Let $f:\{-1,1\}^k\rightarrow\{0,1\}$ be a Boolean function, and $\Psi$ be an instance of $\mcsp(f)$. Then \[\val_{\Psi}\ge \rho(f) + B_{\veclambda}(\Psi) \delta(\Psi)  \, ,\] where $\veclambda=(\widehat{f}(\{1\}),\dots,\widehat{f}(\{k\}))$ and $\delta(\Psi) = \min\{\frac{1}{3k},\frac{2B_{\veclambda}(\Psi)}{9\rho(f)k^2}\}$.
\end{lemma}

\begin{lemma}[Upper bound on $\val_{\Psi}$]\label{lemma:ub}
Let $f:\{-1,1\}^k\rightarrow\{0,1\}$ be a Boolean function, $\epsilon_0(f)$ be as defined in \cref{def:epsilon}, and $\Psi$ be an instance of $\mcsp(f)$. Then
\[\val_{\Psi} \le \frac{B_{\veclambda}(\Psi)+\rho(f)\cdot k}{\epsilon_0(f)+\rho(f)\cdot k}\, ,\] where $\veclambda=(\widehat{f}(\{1\}),\dots,\widehat{f}(\{k\}))$.
\end{lemma}

We defer the proofs of \cref{lemma:lb} and \cref{lemma:ub} to \cref{sec:lemma_proofs}. We now show the correctness of \cref{alg:Bias} using these lemmas. 

\subsubsection{Proof of Theorem~\ref{thm:approximablefn}}
\begin{proof}[Proof of~\cref{thm:approximablefn}]
First, by \cref{prop:ell1norm}, with probability at least $2/3$, $\tilde{B}$ is a $(1\pm\epsilon')$ approximation to $B_{\veclambda}(\Psi)$, i.e., $(1-\epsilon')B_{\veclambda}(\Psi) \leq \tilde{B} \leq (1+\epsilon')B_{\veclambda}(\Psi)$. Next, we show that with probability at least $2/3$, (i) $v\leq \val_{\Psi}$ and (ii) $v \geq \left( \rho(f)+\epsilon^*(f) -\epsilon \right)\cdot\val_\Psi$.
\paragraph{(i) $\bm{v\leq\val_\Psi}$.}
We have \[v= \rho(f) + \frac{\tilde{B}\tilde{\delta}}{(1+\epsilon')^2} \le \rho(f) +  B_{\veclambda}(\Psi) \delta(\Psi)\le \val_{\Psi}\, ,\] where the last inequality follows from \cref{lemma:lb}.

\paragraph{(ii) $\bm{v \geq \left( \rho(f)+\epsilon^*(f) -\epsilon \right)\cdot\val_\Psi}$.}
We have
\begin{equation}\label{eqn:lbv1}
    v = \rho(f) + \frac{\tilde{B}\tilde{\delta}}{(1+\epsilon')^2} \ge \rho(f) + B_{\veclambda}(\Psi) \delta(\Psi) \left(\frac{1-\epsilon'}{1+\epsilon'}\right)^2 \ge \rho(f) + B_{\veclambda}(\Psi) \delta(\Psi) (1-\epsilon)\, ,
\end{equation}
where the last inequality follows from the choice of $\epsilon'$. Let us first consider the case when $B_{\veclambda}(\Psi)\ge \epsilon_0(f)$. We have
\begin{equation}\label{eqn:lbv2}
    B_{\veclambda}(\Psi) \delta(\Psi) \ge \epsilon_0(f)\cdot \min\left\{\frac{1}{3k},\frac{2\epsilon_0(f)}{9\rho(f)k^2}\right\} \ge \epsilon^* \, ,
\end{equation}
where the last equality follows from the definition of $\epsilon^*(f)$ in \cref{def:epsilon}.

Combining \cref{eqn:lbv1} and \cref{eqn:lbv2}, we get 
\[v \ge \rho(f) + \epsilon^*(f) (1-\epsilon) \ge (\rho(f) + \epsilon^*(f) -\epsilon) \val_{\Psi} \, , \] where the last inequality follows from $\val_{\Psi}\le 1$.

Now, let us consider the case when $B_{\veclambda}(\Psi) <  \epsilon_0(f)$.  It follows from \cref{prop:Boolean_fn_prop} that $\epsilon_0(f)\le \rho(f) k$. Therefore, \[\frac{2 B_{\veclambda}(\Psi)}{9\rho(f)k^2} \le \frac{2\epsilon_0(f)}{9\rho(f)k^2} \le \frac{2}{9k} < \frac{1}{3k}\, ,\] and so $\delta(\Psi) = \frac{2 B_{\veclambda}(\Psi)}{9\rho(f)k^2} $. Combining \cref{eqn:lbv1} and \cref{lemma:ub}, we have
\[\frac{v}{\val_{\Psi}}\ge (1-\epsilon)\left(\frac{\rho(f)+\frac{2B_{\veclambda}(\Psi)^2}{9\rho(f)k^2}}{\rho(f)+\frac{B_{\veclambda}(\Psi)}{k}}\right) \left(\rho(f)+\frac{\epsilon_0(f)}{k}\right)\, .\]
We show that for $0\le B_{\veclambda}(\Psi)\le \epsilon_0(f)$, \begin{equation}\label{eqn:minima}\frac{\rho(f)+\frac{2B_{\veclambda}(\Psi)^2}{9\rho(f)k^2}}{\rho(f)+\frac{B_{\veclambda}(\Psi)}{k}} \ge \frac{\rho(f)+\frac{2\epsilon_0(f)^2}{9\rho(f)k^2}}{\rho(f)+\frac{\epsilon_0(f)}{k}} \, .\end{equation} This immediately implies that
\[\frac{v}{\val_{\Psi}}\ge (1-\epsilon) \left(\rho(f)+\frac{2\epsilon_0(f)^2}{9\rho(f)k^2}\right) \ge (1-\epsilon) (\rho(f)+\epsilon^*(f)) > \rho(f)+\epsilon^*(f) -\epsilon \, .\]
Consider the function $g(p) = \frac{\rho(f)+\frac{2p^2}{9\rho(f)}}{\rho(f)+p}$. In order to show \cref{eqn:minima}, it suffices to show that in the range $p\in [0,\frac{\epsilon_0(f)}{k}]$, $g(p)$ attains the minimum value at $p=\frac{\epsilon_0(f)}{k}$, i.e, $g'(p)<0$ in this range. We have
$g'(p)= \frac{\left(\frac{2(p+\rho(f))^2}{9\rho(f)} - \frac{11\rho(f)}{9}\right)}{(\rho(f)+p)^2}$ and for $p\in [0,\frac{\epsilon_0(f)}{k}]$, we have \[\left(\frac{2(p+\rho(f))^2}{9\rho(f)} - \frac{11\rho(f)}{9}\right)\le \left(\frac{2(\epsilon_0(f)/k+\rho(f))^2}{9\rho(f)} - \frac{11\rho(f)}{9}\right) \le \frac{8\rho(f)}{9} - \frac{11\rho(f)}{9} = -\frac{\rho(f)}{3} < 0 \, .\] This completes the proof of~\cref{thm:approximablefn}.
\end{proof}

\subsection{Proofs of Lemma~\ref{lemma:lb} and Lemma~\ref{lemma:ub}}\label{sec:lemma_proofs}
In this section, we prove \cref{lemma:lb} and \cref{lemma:ub}.

\begin{proof}[Proof of \cref{lemma:lb}]
Let $\bern(p)\in \Delta(\{-1,1\})$ denote the Bernoulli distribution where $1$ is sampled with probability $p$. Given an instance $\Psi = (C_1,\ldots,C_m; w_1,\ldots,w_m)$ of $\maxf$ where $C_i = (\vecj(i),\vecb(i))$ and $w_i \geq 0$, let $\gamma = 3\cdot\delta(\Psi) =\min \{\frac{1}{k},\frac{2 B_{\veclambda}(\psi)}{3\rho(f)k^2}\}$. Let $\vecsigma = \arg\max_{\veca \in \{-1,1\}^n} \langle \veca, \bias_{\veclambda}(\Psi)\rangle$. It follows from \cref{lem:bias} that $B_{\veclambda}(\Psi) = \langle \vecsigma, \bias_{\veclambda}(\Psi) \rangle$. In order to prove the lemma, we will show that \[\mathbb{E}_{\veca\sim\left(\bern(\frac{1+\gamma}{2})\right)^n}[\val_\Psi(\veca\odot \vecsigma)] \ge \rho(f) + B_{\veclambda}(\Psi) \delta(\Psi)\, .\] The lemma then directly follows from the fact that $\val_\Psi \ge \mathbb{E}_{\veca\sim\left(\bern(\frac{1+\gamma}{2})\right)^n}[\val_\Psi(\veca\odot \vecsigma)] $. 

We have
\begin{align*}
\mathbb{E}_{\veca\sim\left(\bern(\frac{1+\gamma}{2})\right)^n}[\val_\Psi(\veca\odot \vecsigma)] 
&= \mathbb{E}_{\veca\sim\left(\bern(\frac{1+\gamma}{2})\right)^n}\left[\frac{1}{W}\sum_{i=1}^m w_i \cdot f(\veca|_{\vecj(i)}\odot\vecsigma|_{\vecj(i)}\odot \vecb(i))\right]\\
&= \mathbb{E}_{\veca\sim\left(\bern(\frac{1+\gamma}{2})\right)^n}\left[\frac{1}{W}\sum_{i=1}^m w_i \cdot \sum_{S\subseteq [k]}\widehat{f}(S) \cdot \chi_S(\veca|_{\vecj(i)}\odot\vecsigma|_{\vecj(i)}\odot \vecb(i))\right]\\
&~~~~~~~~~~~~~~~~~~~~~~~~~\mbox{(Fourier expansion of $f$)}\\
&= \frac{1}{W}\sum_{i=1}^m w_i \cdot \sum_{S\subseteq [k]}\widehat{f}(S) \cdot \mathbb{E}_{\veca\sim\left(\bern(\frac{1+\gamma}{2})\right)^n}\left[\chi_S(\veca|_{\vecj(i)}\odot\vecsigma|_{\vecj(i)}\odot \vecb(i))\right]\\
&~~~~~~~~~~~~~~~~~~~~~~~~~\mbox{(Linearity of expectation)}\\
&= \frac{1}{W}\sum_{i=1}^m w_i \cdot \sum_{S\subseteq [k]}\widehat{f}(S) \cdot \chi_S(\vecb(i)\odot\vecsigma|_{\vecj(i)}) \cdot \mathbb{E}_{\veca\sim\left(\bern(\frac{1+\gamma}{2})\right)^n}\left[\chi_S(\veca|_{\vecj(i)})\right]\\
&~~~~~~~~~~~~~~~~~~~~~~~~~\mbox{(Since $\chi_S(a\odot b) = \chi_S(a) \cdot \chi_S(b)$)}\\
&= \frac{1}{W}\sum_{i=1}^m w_i \cdot \sum_{S\subseteq [k]}\widehat{f}(S) \cdot \chi_S(\vecb(i)\odot\vecsigma|_{\vecj(i)}) \cdot \gamma^{|S|}\\
&~~~~~~~~~~~~~~~~~~~~~~~~~\mbox{(Since $\mathbb{E}_{\veca\sim \left(\bern(\frac{1+\gamma}{2})\right)^n}[a_\ell]=\gamma$ for all $\ell\in [n]$)}\\
&= \widehat{f}(\emptyset) +  \frac{1}{W} \sum_{i\in [m]} w_i \sum_{t\in [k]} \widehat{f}(\{t\}) \cdot b(i)_t \cdot \sigma_{j(i)_t}\cdot \gamma \\
&~~~+ \frac{1}{W}\sum_{i=1}^m w_i \cdot \sum_{S\subseteq [k]:|S|\ge 2}\widehat{f}(S) \cdot \chi_S(\vecb(i)\odot\vecsigma|_{\vecj(i)}) \cdot \gamma^{|S|} \\
&= \widehat{f}(\emptyset) + \sum_{\ell \in [n]} \left(\frac{1}{W}\sum_{i\in[m],t\in[k]:j(i)_t=\ell} \widehat{f}(\{t\}) \cdot w_i \cdot b(i)_t \right) \sigma_\ell \cdot \gamma\\
&~~~+ \frac{1}{W}\sum_{i=1}^m w_i \cdot \sum_{S\subseteq [k]:|S|\ge 2}\widehat{f}(S) \cdot \chi_S(\vecb(i)\odot\vecsigma|_{\vecj(i)}) \cdot \gamma^{|S|}\\
&~~~~~~~~~~~~~~~~~~~~~~~~~\mbox{(Rearranging the summations)}\\
&\ge \widehat{f}(\emptyset) + \gamma \langle \bias_{\veclambda}(\Psi),\vecsigma \rangle - \sum_{S\subseteq [k]:|S|\ge 2}|\widehat{f}(S)|\cdot \gamma^{|S|}\\
&~~~~~~~~~~~~~~~~~~~~~~~~~\mbox{(By the definition of $\veclambda$, $\bias_{\veclambda}(\Psi)$), and $|\chi_S(\cdot)|\le 1$)}\\
&\ge \rho(f) + \gamma \cdot B_{\veclambda}(\Psi) - \rho(f) \sum_{S\subseteq [k]:|S|\ge 2}\gamma^{|S|}\\
&~~~~~~~~~~~~~~~~~~~~~~~~~\mbox{(By the definition of $\vecsigma$ and \cref{prop:Boolean_fn_prop})}\\
&= \rho(f) + \gamma \cdot B_{\veclambda}(\Psi) - \rho(f)\sum_{r=2}^k \binom{k}{r}\cdot \gamma^r \, .\\
\end{align*}

We now prove that $\rho(f)\sum_{r=2}^k \binom{k}{r}\cdot \gamma^r \le \frac{2\gamma}{3}\cdot B_{\veclambda}(\Psi)$. Consider the combinatorial identity $\binom{k}{r} = \frac{k\cdot \binom{k-1}{r-1}}{r}$. Since $\gamma \le \frac{1}{k}$ and $r\ge 2$, we have \[\binom{k}{r}\gamma^{r} =\frac{k\cdot \binom{k-1}{r-1}}{r} \cdot \gamma^{r} \le \frac{1}{2}\cdot \binom{k-1}{r-1}\cdot \gamma^{r-1} < \frac{1}{2}\cdot \binom{k}{r-1}\cdot \gamma^{r-1}\, . \] Hence $\sum_{r=2}^k \binom{k}{r}\cdot \gamma^r\le 2 \cdot \binom{k}{2} \cdot \gamma^2$. Since $\gamma \le \frac{2B_{\veclambda}(\Psi)}{3\rho(f)k^2}$, we have 
\[\rho(f)\cdot\sum_{r=2}^k \binom{k}{r}\cdot \gamma^r  \le \rho(f)\cdot 2 \cdot \binom{k}{2} \cdot \gamma^2 \le \frac{2B_{\veclambda}(\Psi)\cdot \gamma}{3} \, . \] 
Recall that $\gamma=3\delta(\Psi)$. Finally, we conclude that
\[ \val_\Psi \ge \mathbb{E}_{\veca\sim\left(\bern(\frac{1+\gamma}{2})\right)^n}[\val_\Psi(\veca\odot \vecsigma)] \ge \rho(f) + \frac{\gamma}{3} \cdot B_{\veclambda}(\Psi) = \rho(f) + B_{\veclambda}(\Psi)\delta(\Psi) \, .\]
\end{proof}

\begin{proof}[Proof of \cref{lemma:ub}]

Let $\Psi = (C_1,\ldots,C_m; w_1,\ldots,w_m)$ be an instance of $\maxf$ where $C_i = (\vecj(i),\vecb(i))$ and $w_i \geq 0$. Let $\veca^*\in \{-1,1\}^n $ denote the assignment that satisfies the maximum weight of constraints in $\Psi$, i.e., $\veca^* = \arg\max_{\veca\in \{-1,1\}^n} \val_\Psi(\veca)$. It follows from \cref{lem:bias} that $B_{\veclambda}(\Psi) \ge \langle \veca^*,\bias_{\veclambda}(\Psi) \rangle$. Let $S$ be the set of indices corresponding to constraints of $\Psi$ satisfied by $\veca^*$, i.e., $S = \{i\in [m]: f(\veca^*|_{\vecj(i)}\odot \vecb(i))=1\}$. We have
\begin{align*}
\langle \veca^*, \bias_{\veclambda}(\Psi)\rangle &= \sum_{\ell\in [n]} a^*_\ell \cdot \frac{1}W \cdot \sum_{i \in [m], t \in [k] : j(i)_t = \ell} \lambda_t w_i b(i)_t\\
&= \frac{1}{W} \sum_{i\in [m]} w_i \sum_{t\in [k]} \lambda_t \cdot b(i)_t \cdot a^*_{j(i)_t}\\
&~~~~~~~~~~~~~~~~~~~~~~~~~\mbox{(Exchanging the summations)} \\
&= \frac{1}{W} \sum_{i\in [m]} w_i \sum_{t\in [k]} \widehat{f}(\{t\}) \cdot b(i)_t \cdot a^*_{j(i)_t}\\
&~~~~~~~~~~~~~~~~~~~~~~~~~\mbox{(~$\veclambda = (\widehat{f}(\{1\}),\dots,\widehat{f}(\{k\}))$~)}\\
&= \frac{1}{W} \sum_{i\in S} w_i \sum_{t\in [k]} \widehat{f}(\{t\}) \cdot b(i)_t \cdot a^*_{j(i)_t} + \frac{1}{W} \sum_{i\notin S} w_i \sum_{t\in [k]} \widehat{f}(\{t\}) \cdot b(i)_t \cdot a^*_{j(i)_t}\\
&\ge \frac{1}{W} \sum_{i\in S} w_i \cdot \epsilon_0(f) - \frac{1}{W} \sum_{i\notin S} w_i \cdot \rho(f) \cdot k\\
&~~~~~~~~~~~~~~~~~~~~~~~~~\mbox{(By the definition of $S$ and $\epsilon_0(f)$, and \cref{prop:Boolean_fn_prop})}\\
&= \val_{\Psi} \cdot \epsilon_0(f) - (1-\val_{\Psi})\rho(f) \cdot k\\
&~~~~~~~~~~~~~~~~~~~~~~~~~\mbox{(By the definition of $\veca^*$)}\, .\\
\end{align*}
Therefore, we get
\[B_{\veclambda}(\Psi) \ge \val_{\Psi} \cdot \epsilon_0(f) - (1-\val_{\Psi})\rho(f)\cdot k \, .\] Rearranging the terms, we get
\[\val_{\Psi} \le \frac{B_{\veclambda}(\Psi)+\rho(f)\cdot k}{\epsilon_0(f)+\rho(f)\cdot k}\, .\]
\end{proof}

\subsection{Approximability of weak monarchy functions}\label{sec:weakmonarchy}
In this section, we analyze the streaming approximability of $\maxf$ where $f$ is a weak monarchy function. Note that in order for $\wmon_{k, j}$ to be a balanced LTF, the total number of votes, i.e., $j + k - 1$, needs to be odd. Therefore, we make such assumption throughout the rest of this section.

\begin{lemma}\label{lem:weakMonarchyFn}
For all integers $j \geq 2$ and  $k\geq 7j^3$ such that $k+j$ is even, 
\[\wmon_{k, j}(x) = \sign \left(\sum_{i=1}^k \widehat{\wmon_{k, j}}(\{i\})x_i \right).\]
\end{lemma}

Note that \cref{lem:weakMonarchyFn} along with \cref{lem:approximablefn} directly conclude~\cref{thm:wmon-apx} restated below.

\thmtwo*

\begin{proof}[Proof of \cref{lem:weakMonarchyFn}]
We start by finding the Chow parameters of $\wmon_{k, j}$. As mentioned earlier, we only consider the case where $k + j$ is even. For the president,
\begin{eqnarray*}
\widehat{\wmon_{k, j}}(\{1\}) 
&=& \Pr\{x_1 = 1, \wmon_{k, j}(x) = 1\} \times 1\\
&& \Pr\{x_1 = -1, \wmon_{k, j}(x) = 1\} \times (-1)\\
&& \Pr\{x_1 = 1, \wmon_{k, j}(x) = 0\} \times 0\\
&& \Pr\{x_1 = -1, \wmon_{k, j}(x) = 0\} \times 0\\
&=& \frac{1}{2^{k}}\left(\binom{k-1}{\geq \frac{k+j}{2}-j}
- \binom{k-1}{\geq \frac{k+j}{2}}
\right)\\
&=& \frac{1}{2^{k}}\left(\binom{k-1}{\geq \frac{k-j}{2}}
- \binom{k-1}{\geq \frac{k+j}{2}}
\right)\\
&=& \frac{1}{2^{k}}\left(
2^{k-1} - 2\binom{k-1}{<\frac{k-j}{2}}
\right) \, .
\end{eqnarray*}

For citizen $x_i$ ($i > 1$),
\begin{eqnarray*}
\widehat{\wmon_{k, j}}(\{i\}) &=& \Pr\{x_1 = 1, x_i = 1,  \wmon_{k, j}(x) = 1\} \times 1\\
&& \Pr\{x_1 = 1, x_i = -1,  \wmon_{k, j}(x) = 1\} \times (-1)\\
&& \Pr\{x_1 = -1, x_i = 1,  \wmon_{k, j}(x) = 1\} \times 1\\
&& \Pr\{x_1 = -1, x_i = -1,  \wmon_{k, j}(x) = 1\} \times (-1)\\
&& \Pr\{x_1 = 1, x_i = 1,  \wmon_{k, j}(x) = 0\} \times 0\\
&& \Pr\{x_1 = 1, x_i = -1,  \wmon_{k, j}(x) = 0\} \times 0\\
&& \Pr\{x_1 = -1, x_i = 1,  \wmon_{k, j}(x) = 0\} \times 0\\
&& \Pr\{x_1 = -1, x_i = -1,  \wmon_{k, j}(x) = 0\} \times 0\\
&=& \frac{1}{2^{k}}\left(
\binom{k-2}{\geq \frac{k+j}{2}-j-1}
- \binom{k-2}{\geq \frac{k+j}{2}-j}
+ \binom{k-2}{\geq \frac{k+j}{2}-1}
- \binom{k-2}{\geq \frac{k+j}{2}}
\right)\\
&=& \frac{1}{2^{k}}\left(
\binom{k-2}{\frac{k+j}{2}-j-1}
+ \binom{k-2}{\frac{k+j}{2}-1}
\right)\\
&=& \frac{\binom{k-2}{\frac{k-j}{2}-1}}{2^{k-1}} \, .
\end{eqnarray*}

Note that in order for functions $\wmon_{k, j}(x)$ and $\sign \left(\sum_{i=1}^k \widehat{\wmon_{k, j}}(\{i\})x_i\right)$ to be the same, it suffices to have
$$
j - 1 < \frac{\widehat{\wmon_{k, j}}(\{1\})}{\widehat{\wmon_{k, j}}(\{i\})}< j + 1 \, .\footnote{Indeed, letting $w=\frac{\widehat{\wmon_{k, j}}(\{1\})}{\widehat{\wmon_{k, j}}(\{i\})}$, if $j-1<w<j+1$, then we have that $\sign \left(\sum_{i=1}^k \widehat{\wmon_{k, j}}(\{i\})x_i\right)=\sign \left(wx_1+\sum_{i=2}^k x_i\right)=\sign\left(j x_1 +\sum_{i=2}^k x_i\right)=\wmon_{k, j}(x)$.}
$$
Thus, in the rest of the proof, we find values for $k$ that guarantee the bounds above. We start with the upper-bound:

\begin{equation}
\frac{\widehat{\wmon_{k, j}}(\{1\})}{\widehat{\wmon_{k, j}}(\{i\})} 
= \frac{2^{k-2} - \binom{k-1}{<\frac{k-j}{2}}}{\binom{k-2}{\frac{k-j}{2}-1}}
\leq \frac{
\frac{j}{2} \cdot \binom{k-1}{\lfloor\frac{k-1}{2}\rfloor}
}{\binom{k-2}{\frac{k-j}{2}-1}} \, .
\end{equation}

The last inequality holds as below:
\begin{itemize}
    \item If $k-1$ is odd: $2^{k-2} = \sum_{i=0}^{\frac{k-2}{2}} \binom{k-1}{i}
    \Rightarrow 
    2^{k-2} - \binom{k-1}{<\frac{k-j}{2}} =\sum_{i=\frac{k-j}{2}}^{\frac{k-2}{2}} \binom{k-1}{i} \leq \frac{j}{2} \cdot \binom{k-1}{\frac{k-2}{2}}$
    
    \item If $k-1$ is even: $2^{k-2} = \frac{1}{2}\binom{k-1}{\frac{k-1}{2}}+\sum_{i=0}^{\frac{k-3}{2}} \binom{k-1}{i}
    \Rightarrow 
    2^{k-2} - \binom{k-1}{<\frac{k-j}{2}} \leq \frac{j}{2} \cdot \binom{k-1}{\frac{k-1}{2}}$
\end{itemize}

Therefore, 
\begin{eqnarray*}
\frac{\widehat{\wmon_{k, j}}(\{1\})}{\widehat{\wmon_{k, j}}(\{i\})} 
&\leq& \frac{j}{2} \cdot \frac{\frac{(k-1)!}{(\lfloor\frac{k-1}{2}\rfloor)!(\lceil\frac{k-1}{2}\rceil)!}}
{\frac{(k-2)!}{(\frac{k-j}{2}-1)!(\frac{k+j}{2}-1)!}}
= \frac{j}{2}\cdot \frac{k-1}{\lfloor\frac{k-1}{2}\rfloor} \cdot \frac{(\frac{k+j}{2}-1) \cdots (\lceil\frac{k-1}{2}\rceil + 1) }{(\lfloor\frac{k-1}{2}-1\rfloor) \cdots (\frac{k-j}{2})}\\
&\leq& \frac{j}{2} \cdot 2\left(1 + \frac{1}{k-2}\right) \cdot \left(\frac{\lceil\frac{k-1}{2}\rceil + 1}{\frac{k-j}{2}}\right)^{\lfloor\frac{j - 1}{2}\rfloor}\\
&\leq& j \cdot \left(1 + \frac{1}{k-2}\right) \cdot 
\left(\frac{k + 2}{k-j}\right)^{\frac{j - 1}{2}}\\
&=& j \cdot \left(1 + \frac{1}{k-2}\right) \cdot 
\left(1 + \frac{j + 2}{k-j}\right)^{\frac{j - 1}{2}}
\leq j \cdot \left(1 + \frac{j + 2}{k-j}\right)^{j} \, .
\end{eqnarray*}
For any given $j$, $\left(1 + \frac{j + 2}{k-j}\right)^{j}$ tends to 1 as $k$ goes to $\infty$. Therefore, there exits some $K_0$ such that for all $k \geq K_0$, 
$ \frac{\widehat{\wmon_{k, j}}(\{1\})}{\widehat{\wmon_{k, j}}(\{i\})}  < j + 1.$ More precisely, we take $k$ to be at least $K_0 = 2j^3 + 4j^2 + j \leq 7j^3$. This way, 

\begin{eqnarray*}
\frac{\widehat{\wmon_{k, j}}(\{1\})}{\widehat{\wmon_{k, j}}(\{i\})} 
&\leq&  j \cdot \left(1 + \frac{1}{2j^2}\right)^{j}
\leq  j \cdot \left(1
+ j\cdot \frac{1}{2j^2}
+ j^2\cdot \frac{1}{(2j^2)^2}
+ j^3\cdot \frac{1}{(2j^2)^3}
+ \cdots
\right)\\
&=&  j \cdot \left(1
+ \frac{1}{2j}
+ \frac{1}{(2j)^2}
+ \frac{1}{(2j)^3}
+ \cdots
\right)\\
&<& j \cdot \left(1 + \frac{1}{j}\right)
.
\end{eqnarray*}

We now proceed to the lower bound. 
\begin{equation}
\frac{\widehat{\wmon_{k, j}}(\{1\})}{\widehat{\wmon_{k, j}}(\{i\})} 
= \frac{2^{k-2} - \binom{k-1}{<\frac{k-j}{2}}}{\binom{k-2}{\frac{k-j}{2}-1}}
\geq \frac{
\frac{j}{2} \cdot \binom{k-1}{\frac{k-j}{2}}
}{\binom{k-2}{\frac{k-j}{2}-1}} \, .
\end{equation}

Similar to the upper-bound case, the last inequality can be observed as follows:
\begin{itemize}
    \item If $k-1$ is odd: $2^{k-2} = \sum_{i=0}^{\frac{k-2}{2}} \binom{k-1}{i}
    \Rightarrow 
    2^{k-2} - \binom{k-1}{<\frac{k-j}{2}} =\sum_{i=\frac{k-j}{2}}^{\frac{k-2}{2}} \binom{k-1}{i} \geq \frac{j}{2} \cdot \binom{k-1}{\frac{k-j}{2}}$
    
    \item If $k-1$ is even: $2^{k-2} = \frac{1}{2}\binom{k-1}{\frac{k-1}{2}}+\sum_{i=0}^{\frac{k-3}{2}} \binom{k-1}{i}
    \Rightarrow 
    2^{k-2} - \binom{k-1}{<\frac{k-j}{2}} \geq \frac{j}{2} \cdot \binom{k-1}{\frac{k-j}{2}}$
\end{itemize}

Therefore,
\begin{eqnarray*}
\frac{\widehat{\wmon_{k, j}}(\{1\})}{\widehat{\wmon_{k, j}}(\{i\})} 
&\geq& \frac{j}{2} \cdot \frac{
 \binom{k-1}{\frac{k-j}{2}}
}{\binom{k-2}{\frac{k-j}{2}-1}} 
= \frac{j}{2} \cdot \frac{
\binom{k-2}{\frac{k-j}{2}}+\binom{k-2}{\frac{k-j}{2}-1}
 }{\binom{k-2}{\frac{k-j}{2}-1}} 
 = \frac{j}{2} \cdot \left(1 + \frac{
\binom{k-2}{\frac{k-j}{2}} }{\binom{k-2}{\frac{k-j}{2}-1}} \right)
 \\
 &=&\frac{j}{2} \cdot \left(1 + \frac{\frac{k+j}{2}-1}{\frac{k-j}{2}} \right)
 =\frac{j}{2} \cdot \left(1 + \frac{k+j-2}{k-j} \right)\\
 &=&j \cdot \left(1 + \frac{j-1}{k-j} \right) \, .
\end{eqnarray*}

This lower bound is larger than $j$ for every $k > j$. Thus, for every $k \geq 2j^3 + 4j^2 + j$, 
$j \leq \frac{\widehat{\wmon_{k, j}}(\{1\})}{\widehat{\wmon_{k, j}}(\{i\})} < j + 1$ which implies that 
$\wmon_{k, j}(x) = \sign \left(\sum_{i=1}^k \widehat{\wmon_{k, j}}(\{i\})x_i \right)$, and concludes the proof.
\end{proof}

\section*{Acknowledgments}
We thank the anonymous reviewers for their helpful and constructive comments.

\bibliography{references}

\newcommand{\etalchar}[1]{$^{#1}$}
\begin{thebibliography}{BHP{\etalchar{+}}22}

\bibitem[ABM10]{abm10}
Per Austrin, Siavosh Benabbas, and Avner Magen.
\newblock On quadratic threshold csps.
\newblock In {\em LATIN 2010}, pages 332--343. Springer, 2010.

\bibitem[BHP{\etalchar{+}}22]{BHP+22}
Joanna Boyland, Michael Hwang, Tarun Prasad, Noah Singer, and Santhoshini
  Velusamy.
\newblock Closed-form expressions for the sketching approximability of (some)
  symmetric {Boolean CSPs}.
\newblock {\em CoRR}, abs/2112.06319, February 2022.

\bibitem[CGS{\etalchar{+}}22]{CGS+22}
Chi-Ning Chou, Alexander Golovnev, Madhu Sudan, Ameya Velingker, and
  Santhoshini Velusamy.
\newblock Linear {{Space Streaming Lower Bounds}} for {{Approximating CSPs}}.
\newblock In {\em STOC 2022}, 2022.
\newblock To appear.

\bibitem[CGSV21]{CGSV21-finite}
Chi-Ning Chou, Alexander Golovnev, Madhu Sudan, and Santhoshini Velusamy.
\newblock Approximability of all finite {{CSPs}} with linear sketches.
\newblock In {\em FOCS 2021}, pages 1197--1208. {IEEE}, 2021.

\bibitem[CGSV22]{CGSV21-boolean}
Chi{-}Ning Chou, Alexander Golovnev, Madhu Sudan, and Santhoshini Velusamy.
\newblock Approximability of all {{Boolean CSPs}} with linear sketches.
\newblock {\em CoRR}, abs/2102.12351v8, 11th February 2022.

\bibitem[CGV20]{CGV20}
Chi-Ning Chou, Alexander Golovnev, and Santhoshini Velusamy.
\newblock Optimal {{Streaming Approximations}} for all {{Boolean Max-2CSPs}}
  and {{Max-}}{\(k\)}{{SAT}}.
\newblock In {\em FOCS 2020}, pages 330--341. {IEEE}, 2020.

\bibitem[GT19]{GT19}
Venkatesan Guruswami and Runzhou Tao.
\newblock Streaming {{Hardness}} of {{Unique Games}}.
\newblock In {\em APPROX 2019}, pages 5:1--5:12. {Schloss Dagstuhl}, 2019.

\bibitem[GVV17]{GVV17}
Venkatesan Guruswami, Ameya Velingker, and Santhoshini Velusamy.
\newblock Streaming {{Complexity}} of {{Approximating Max 2CSP}} and {{Max
  Acyclic Subgraph}}.
\newblock In {\em APPROX 2017}, pages 8:1--8:19. {Schloss Dagstuhl}, 2017.

\bibitem[Has05]{h05}
Gustav Hast.
\newblock {\em Beating a random assignment: Approximating constraint
  satisfaction problems}.
\newblock PhD thesis, KTH, 2005.

\bibitem[HP20]{hp20}
Neng Huang and Aaron Potechin.
\newblock On the approximability of presidential type predicates.
\newblock In {\em APPROX 2020}, pages 58:1--58:20. Schloss Dagstuhl, 2020.

\bibitem[Ind00]{Indyk}
Piotr Indyk.
\newblock Stable distributions, pseudorandom generators, embeddings and data
  stream computation.
\newblock In {\em FOCS 2000}, pages 189--197. IEEE, 2000.

\bibitem[KK15]{KK15}
Dmitry Kogan and Robert Krauthgamer.
\newblock Sketching cuts in graphs and hypergraphs.
\newblock In {\em ITCS 2015}, pages 367--376. ACM, 2015.

\bibitem[KK19]{KK19}
Michael Kapralov and Dmitry Krachun.
\newblock An optimal space lower bound for approximating {{MAX-CUT}}.
\newblock In {\em STOC 2019}, pages 277--288. ACM, 2019.

\bibitem[KKS15]{KKS15}
Michael Kapralov, Sanjeev Khanna, and Madhu Sudan.
\newblock Streaming lower bounds for approximating {{MAX-CUT}}.
\newblock In {\em SODA 2015}, pages 1263--1282. SIAM, 2015.

\bibitem[KKSV17]{KKSV17}
Michael Kapralov, Sanjeev Khanna, Madhu Sudan, and Ameya Velingker.
\newblock {\((1 + \omega(1))\)}-approximation to {{MAX-CUT}} requires linear
  space.
\newblock In {\em SODA 2017}, pages 1703--1722. SIAM, 2017.

\bibitem[KNW10]{KNW10}
Daniel~M. Kane, Jelani Nelson, and David~P. Woodruff.
\newblock On the exact space complexity of sketching and streaming small norms.
\newblock In {\em SODA 2010}, pages 1161--1178. SIAM, 2010.

\bibitem[O'D14]{ryansbook}
Ryan O'Donnell.
\newblock {\em Analysis of {B}oolean functions}.
\newblock Cambridge University Press, 2014.

\bibitem[Pot19]{Pot19}
Aaron Potechin.
\newblock On the approximation resistance of balanced linear threshold
  functions.
\newblock In {\em STOC 2019}, pages 430--441. ACM, 2019.

\bibitem[SSV21]{SSV21}
Noah Singer, Madhu Sudan, and Santhoshini Velusamy.
\newblock Streaming approximation resistance of every ordering {{CSP}}.
\newblock In {\em APPROX 2021}, pages 17:1--17:19. {Schloss Dagstuhl}, 2021.

\end{thebibliography}
\bibliographystyle{alpha}
\end{document}